\documentclass[11pt]{amsart}
\usepackage{amsmath}
\usepackage{graphicx}
\usepackage{bm}
\usepackage{subfig,tikz}
\usepackage{subfig,tikz}
\usetikzlibrary{decorations.markings}
\usepackage{wrapfig}
\usepackage{xspace}
\usepackage{enumerate}
\theoremstyle{plain}
\newtheorem{theorem}{Theorem}[section]
\newtheorem{lemma}[theorem]{Lemma}

\newtheorem{cor}[theorem]{Corollary}

\newtheorem{proposition}[theorem]{Proposition}

\theoremstyle{definition}
\newtheorem{remark}[theorem]{Remark}

\numberwithin{equation}{section}

\newcommand{\tr}{\operatorname{Tr}}

\begin{document}
	

	

\title[Existence/Uniqueness in 5-Dimensional Minimal Supergravity]{Existence and Uniqueness of Stationary Solutions in 5-Dimensional Minimal Supergravity}

\author{Aghil Alaee}
\address{Center of Mathematical Sciences and Applications, Harvard University,
20 Garden Street, Cambridge, MA 02138, USA}
\address{Department of Mathematics and Computer Science,
Clark University, Worcester, MA 01610, USA}
\email{aghil.alaee@cmsa.fas.harvard.edu, aalaeekhangha@clarku.edu}	

		
\author{Marcus Khuri}
\address{Department of Mathematics, Stony Brook University, Stony Brook, NY 11794, USA}
\email{khuri@math.sunysb.edu}
	
\author{Hari Kunduri}
\address{Department of Mathematics and Statistics,
Memorial University of Newfoundland,
St John's NL A1C 4P5, Canada}
\email{hkkunduri@mun.ca}
	
	
\thanks{A. Alaee acknowledges the support of a NSERC Postdoctoral Fellowship. M. Khuri acknowledges the support of NSF Grant DMS-2104229, and Simons Foundation Fellowship 681443. H. Kunduri acknowledges the support of NSERC Discovery Grant RGPIN-2018-04887.}

\begin{abstract}
We study the problem of stationary bi-axially symmetric solutions of the $5$-dimensional minimal supergravity equations. Essentially all possible solutions with nondegenerate horizons are produced, having the allowed horizon cross-sectional topologies of the sphere $S^3$, ring $S^1\times S^2$, and lens $L(p,q)$, as well as the three different types of asymptotics.  The solutions are smooth apart from possible conical singularities at the fixed point sets of the axial symmetry. This analysis also includes the solutions known as solitons in which horizons are not present but are rather replaced by nontrivial topology called bubbles which are sustained by dipole fluxes. Uniqueness results are also presented which show that the solutions are completely determined by their angular momenta, electric and dipole charges, and rod structure which fixes the topology. Consequently we are able to identify the finite number of parameters
that govern a solution. In addition, a generalization of these results is given where the spacetime is allowed to have orbifold singularities.
\end{abstract}

\maketitle


\section{Introduction}
\label{sec1}\setcounter{equation}{0}
\setcounter{section}{1}

A foundational result in mathematical relativity is the proof that the domain of outer communications of any stationary and axisymmetric asymptotically flat black hole solution with a connected, non-degenerate horizon of the Einstein-Maxwell system is isometric to the domain of outer communications of the three-parameter  Kerr-Newman family of solutions \cite[Theorem 3.2]{Chrusciel:2012jk}.  The first step was Hawking's observation that each horizon cross-section of a (not necessarily stationary) black hole must have the topology of a sphere $S^2$ \cite{HawkingEllis}; the exceptional case of a torus $T^2$ was ruled out in \cite{Galloway:2006ws}. As originally shown in \cite{Israel:1967za} under a set of restrictive assumptions,  the Reissner-Nordstr\"{o}m family exhausts the set of static black holes in this class . The result was subsequently strengthened under far weaker conditions \cite{Bunting:1987,Masood,Ruback} (in particular horizons with multiple components were ruled out).  Finally, it was observed that the Einstein-Maxwell equations reduced on stationary, axisymmetric solutions are equivalent to a harmonic map from a half plane to the 4-dimensional complex hyperbolic space $SU(2,1)/S(U(2)\times U(1))$ \cite{Carter:1973rla}.  This reduces the classification problem to that of a 2-dimensional elliptic (singular) boundary value problem. The uniqueness result follows from this, with existence guaranteed by the explicit construction of the Kerr-Newman solution \cite{Bunting,Mazur:1982db}. The assumption of axisymmetry can be replaced with that of analyticity of the solution and there is recent work on removing the latter assumption \cite{Alexakis:2009ch}.

The classification problem for stationary asymptotically flat black hole solutions in dimensions greater than four is of intrinsic interest, as it is clear that higher-dimensional general relativity has a number of novel features that distinguish it from the standard $D=4$ setting \cite{Emparan:2008eg}. In addition, string theory, the leading candidate for a theory of quantum gravity, asserts the existence of more than three spatial dimensions.  In phenomenological models, a subset of these dimensions are `compactified' (i.e. considered very small) and dynamics in the remaining macroscopic dimensions is governed by supergravity 
theories. These are extensions of general relativity with additional scalar fields and both Abelian and non-Abelian gauge fields.  Black holes arise naturally in this context, and indeed a major success of string theory is a quantum mechanical account of the Bekenstein-Hawking entropy of a certain class of degenerate (extreme) black holes \cite{Strominger:1996sh}.

The horizon topology theorem has been established in $D>4$ by Galloway and Schoen \cite{Galloway:2006ws,Galloway:2005mf}. They prove that a black hole solution satisfying the dominant energy condition must have a horizon cross-section with positive Yamabe invariant, that is it admits a metric with positive scalar curvature. In dimensions $D=5$ this shows that the possible horizon topologies are the sphere $S^3$ (or more generally, a space covered by $S^3$, such as lens spaces $L(p,q)$), $S^1 \times S^2$ and connected sums thereof.  Explicit asymptotically flat, stationary bi-axisymeetric supergravity solutions corresponding to $S^3$ (the charged Myers-Perry family \cite{Cvetic:1996xz}), $S^1 \times S^2$ (charged black rings \cite{Elvang:2004xi, Elvang:2004rt})  and $L(p,1)$ \cite{Kunduri:2014kja, Tomizawa:2017suc} are known. We note that in the vacuum case, asymptotically flat black hole solutions with $L(p,q)$ horizons have been produced abstractly but they have not  yet been shown to be devoid of conical singularities \cite{KhuriWeinsteinYamada}.  Furthermore, any static electrically charged black hole must belong to the appropriate $D>4$ generalization of the 2-parameter  Reissner-Nordstr\"{o}m family \cite{Gibbons:2002ju}. Generically it is expected that non-static solutions will be cohomogeneity-two or higher (the dimension of the orbit space is greater than 1) making the equations more difficult to analyze systematically. The explicit solutions above have been constructed either by generalizing aspects of the Kerr solution (i.e. the `Kerr-Schild' form), inverse scattering techniques associated with integrability of the field equations \cite{Figueras:2009mc}, or extra geometric constraints satisfied by supersymmetric solutions \cite{Gauntlett:2002nw} such as the existence of Killing spinors. In fact, a classification of stationary bi-axisymmetric supersymmetric solutions has recently been achieved \cite{Breunholder:2017ubu}, and it is known that supersymmetric black holes are not characterized uniquely by their conserved charges computed at spatial infinity \cite{Breunholder:2018roc,Horowitz:2017fyg}. The moduli space of generic non-supersymmetric solutions is clearly quite rich, and given that the techniques above are not systematized, it is highly unlikely that a classification can be achieved with such methods.

A key feature of certain supergravity theories is that upon Kaluza-Klein reduction they can be recast as a sigma model/harmonic map with a symmetric space target manifold. Note that while this property is not satisfied by the Einstein-Maxwell equations for $D=5$ unless additional conditions are imposed on the metric \cite{Hollands:2007qf}, $D=5$ minimal supergravity restricted to stationary biaxisymmetric solutions is equivalent to a harmonic map with 8-dimensional target $G_{2(2)}/ SO(4)$ \cite{Bouchareb:2007ax}.  This structure allows one to establish uniqueness results \cite{Armas:2014gga,Armas:2009dd,Tomizawa:2010xj,Tomizawa:2009ua,Tomizawa:2009tb} (generalizing uniqueness theorems in the vacuum setting \cite{Figueras:2009ci,Hollands:2007aj}), once a given set of geometric invariants is specified. However, these works do not address the problem of existence and most of what is known relies on explicitly constructing solutions as discussed above. More abstract methods of constructing solutions must therefore be employed.  The purpose of this paper is to complete the uniqueness study and give a general existence theory based on the PDE approach developed in \cite{KWetal,KhuriWeinsteinYamada,KhuriWeinsteinYamada1,WeinsteinHadamard} for the vacuum case. The supergravity setting possesses a number of new qualitative features, which we describe below.

A \emph{gravitational soliton} is a nontrivial, globally stationary, geodesically complete spacetime.  Such solutions necessarily do not contain black hole event horizons. It is a classic result of Lichnerowicz \cite{Lich} (see also \cite{Anderson:1999xz} for a more general result) that the vacuum field equations do not admit any asymptotically flat gravitational solitons.    A similar result holds for the Einstein-Maxwell system in $D=4$.  This can be seen more easily in modern terms by using the positive mass theorem \cite{Schon:1979rg,Witten}.  Let $\xi$ denote the stationary Killing vector field and $\Sigma$ be a Cauchy surface. The mass of an asymptotically flat globally hyperbolic spacetime is given by the Komar integral
\begin{equation}
\mathbf{m} = -\frac{(D-2)}{16\pi (D-3)} \int_{S^{D-2}_\infty} \star d\xi = \frac{(-1)^D (D-2)}{8\pi (D-3)} \int_\Sigma  \star \text{Ric}(\xi),
\end{equation}
where $S^{D-2}_\infty$ indicates a limit on coordinates spheres in the asymptotic end. The second equality is obtained via Stokes' theorem, in addition to the fact $\xi$ is Killing and $\Sigma$ does not have an inner boundary. In vacuum, the volume integral vanishes and hence $\mathbf{m}=0$.  We conclude that the spacetime must be Minkowski space by the rigidity statement of the positive mass theorem. In Einstein-Maxwell theory, the field strength satisfies
\begin{equation}
d\mathcal{F} =0, \qquad  d\star \mathcal{F}=0.
\end{equation}
By topological censorship \cite{Friedman:1993ty} the spacetime is simply connected, and hence if $D=4$ there exist globally defined electric and magnetic potentials $d\psi_E= -\iota_\xi \mathcal{F}$, $d\psi_M = -\iota_\xi \star \mathcal{F}$. The Einstein equations then imply that the volume integrand is exact
\begin{equation}
\star \text{Ric}(\xi) = \frac{1}{4} d \left( \psi_E \star \mathcal{F}- \psi_M \mathcal{F}\right),
\end{equation}
and we then have $\mathbf{m}=0$ once again using the fact $\mathcal{F} \to 0$ at spatial infinity.  However, if $D>4$ the homology group $H_{D-3}$ may be nontrivial. In particular for $D=5$ the spacetime may admit nontrivial 2-cycles.  In this case  $\iota_\xi \star \mathcal{F}$ need not be exact, and thus the volume integral above is not necessarily zero. In Einstein-Maxwell theory, one can rule out the existence of \emph{static} solitons for all $D$ \cite{Kunduri:2017htl}, and no stationary examples are known. Physically, these cycles may carry magnetic flux and this energy contributes to a nonzero spacetime mass.

The Maxwell equations of supergravity (\eqref{SFEintro} below) have a nonlinear source term. One may still construct a closed magnetic $D-3$ form as above, but  again nontrivial cycles present an obstruction to the above argument.  Remarkably, there is now a large class of explicit examples of solitons in supergravity; they are also referred to as `smooth geometries' or `fuzzballs', see e.g. the review \cite{Bena:2007kg}.  These solutions  have nonzero charge and angular momenta.  Almost all of the known families are supersymmetric  (e.g. admit solutions of the Killing spinor equations) and satisfy the BPS relation $\mathbf{m} = |\mathcal{Q}|$ in appropriate units.

The existence of nontrivial topology in the domain of outer communications also raises the question: do there exist black holes with 2-cycles in the domain of outer communication? This would present a strong departure from the familiar four-dimensional case (no-hair theorems), in which black holes can be characterized by their asymptotic conserved charges. In particular,  quite different black hole spacetimes containing such 2-cycles could not be distinguished from those without, merely by their mass $\mathbf{m}$, electric charge $\mathcal{Q}$, and angular momenta $\mathcal{J}_i$. Recently, an explicit example of an asymptotically flat supersymmetric black hole with $S^3$ horizon and a 2-cycle in the exterior was found \cite{Kunduri:2014iga}. This construction was also generalized in the analysis of supersymmetric solutions in \cite{Breunholder:2018roc}.  The solution can be interpreted physically as an equilibrium configuration of a black hole and soliton \cite{Horowitz:2017fyg}.  No non-supersymmetric examples are known, although one would expect such solutions to exist at least for a restricted region of the moduli space of solutions. Furthermore, it has been shown that the familiar first law of black hole mechanics must be modified with new nonlinear contributions from solitons in the exterior region~\cite{Kunduri:2013vka}.  In the present article, we will construct the first non-supersymmetric solutions of this type.

\section{Statement of Main Results}

We first recall basic notions associated with
stationary bi-axisymmetric spacetimes. The group of isometries of the spacetime $(M^5, \mathbf{g}, \mathcal{F})$ admits an Abelian subgroup $\mathbb{R} \times U(1) \times U(1)$, the generators of which necessarily commute with each other (if there are additional matter fields, such as a Maxwell field, these fields are assumed to be invariant under these symmetries as well).  The set of points where a closed-orbit Killing field degenerates is an axis, which appears as an interval or \emph{rod} on the $z$-axis in the 2-dimensional orbit space $M^5/[U(1)^2\times\mathbb{R}]$ \cite{Hollands:2012xy} or in the domain $\mathbb{R}^3$ of the harmonic map \cite{KhuriWeinsteinYamada}. Each axis rod $\Gamma_l$ comes equipped with a pair of mutually prime integers $v_l=(v_l^1,v_l^2)$, referred to as the \emph{rod structure}, which indicates the particular linear combination of rotational Killing fields that vanishes on this rod. The entire $z$-axis is decomposed into axis rods and horizon rods, the latter having rod structure $(0,0)$. End points of horizon rods are called \emph{poles}, and points separating two axis rods are called \emph{corners} where both $U(1)$ generators vanish. The collection of rod structures completely determines the topology of the domain of outer communications, as well as that of the horizons. We will seek to produce solutions with a prescribed rod structure, and hence a prescribed topology. An \textit{admissibility condition} is required to prevent orbifold singularities at a corner associated with a given rod structure \cite{Hollands:2012xy}. This states that the determinant \eqref{admiss} formed by the $2\times 2$ matrix of neighboring rod structures at a corner is $\pm 1$. Moreover, for technical reasons tied to the construction of the harmonic map, when three consecutive axis rods $\Gamma_{l-1}$, $\Gamma_l$, and $\Gamma_{l+1}$ are present an additional \emph{compatibility condition} is needed. If the admissibility condition determinant is $+1$, which may be assumed without loss of generality, the compatibility condition asserts that
$v_{l-1}^1 v_{l+1}^1\leq 0$. A \emph{generalized compatibility condition} \eqref{compat1} is utilized in the case that orbifold singularities are present.
It should be pointed out that these conditions do not restrict the possible horizon topologies that can be produced with our approach, which includes all prime 3-manifolds with positive Yamabe invariant.

A \emph{rod data set} consists of the collection of rods, and rod structures $v_l$, together with a prescription of four constants $\mathbf{a}_l\in\mathbb{R}^2$ and $\mathbf{b}_l,\mathbf{c}_l\in \mathbb{R}$ on each axis rod $\Gamma_l$. The values of $\mathbf{a}_l$ and $\mathbf{b}_l$ do not change between rods that share a corner. Thus, these constants may only experience jumps across a horizon rod, and the difference between these constants on each side of a horizon
rod determines the angular momenta $\mathcal{J}_i$, $i=1,2$ and electric charge $\mathcal{Q}$ of the horizon component, see Section \ref{sec5}. The constants $\mathbf{c}_l$ may change across a corner, and their value fixes the dipole charges $\mathcal{D}_l$ of the 2-cycle associated with an axis rod bounded by two corners. 
The collection of such constants may be interpreted as axis boundary data for the potentials used to construct the harmonic map $\varphi\colon\mathbb{R}^3\setminus\Gamma\rightarrow G_{2(2)}/SO(4)$, and the rod structures uniquely determine its prescribed singularities; here $\Gamma$ is the union of all axis rods. The notion of prescribed singularities in this context refers to the prescribed nature of the blow-up of the harmonic map at axis rods, which is enforced by requiring the solutions to be asymptotic to a `model map' that encodes the rod structures. It is important to point out that a primary difference with the vacuum setting \cite{KhuriWeinsteinYamada} is that the potential constants here \textit{do not} agree with the restriction of the potentials to the axes, but rather agree with a nonlinear combination of the potentials on the axes. It is this crucial observation that underlies the difficulty in the minimal supergravity case, and guides the results of this paper. A rod data set having
potential constants as described and satisfying the admissibility condition will
be referred to as admissible.

We will refer to an asymptotically flat stationary vacuum spacetime as \emph{well-behaved} if the stationary Killing field has complete orbits, and the domain of outer communications is globally hyperbolic with
an acausal spacelike connected (Cauchy) hypersurface that is asymptotic to the canonical slice in the asymptotic end, and is such that the boundary (if nonempty) is a compact cross section of the horizon. These conditions are utilized in the dimensional reduction argument and are consistent with those of \cite{Chrusciel:2012jk}. Furthermore, we will only treat \textit{non-degenerate horizons}. The non-degeneracy condition refers to the requirement that the (Killing) horizon surface gravity is nonzero, which in the current setting is equivalent to the horizon rod having nonzero length \cite{Hollands:2012xy}.

When constructing the spacetime from a given harmonic map the issue of conical
singularities along the axis rods arises. Along each axis rod $\Gamma_l$ the angle deficit, consisting of the limiting ratio between $2\pi$ times the radius from $\Gamma_l$ to the orbit of the Killing field degenerating on this rod and the
length of the orbit, may be different from 1. This conical singularity may be
thought of as a `strut' along the axis holding the system in a stationary configuration. In order for a solution to be considered physically relevant, we require the absence of conical singularities along each axis rod. This entails
balancing the various parameters which define the solution.

\begin{theorem} \label{primaryresult}\par
\noindent
\begin{enumerate}[(i)]
\item
A well-behaved stationary bi-axially symmetric asymptotically flat solution of the 5D minimal supergravity equations without degenerate horizons yields a harmonic map $\varphi\colon\mathbb{R}^3\setminus\Gamma\rightarrow G_{2(2)}/SO(4)$ having axis singularities consistent with an admissible rod data set, and is free of conical singularities.
\item
Conversely,
given a rod data set satisfying the admissibility and compatibility conditions, there exists a unique harmonic map $\varphi\colon\mathbb{R}^3\setminus\Gamma\rightarrow G_{2(2)}/SO(4)$ having the prescribed singularities and potential values on $\Gamma$ associated with the rod data.
\item
Given $\varphi$ as in $(ii)$, a well-behaved stationary bi-axially symmetric asymptotically flat solution of the 5D minimal supergravity equations without degenerate horizons can be constructed.
\end{enumerate}
\end{theorem}

\begin{remark}
These results may be extended to the setting of spacetimes which are asymptotically Kaluza-Klein (AKK) and asymptotical locally Euclidean (ALE), as was done in the vacuum case \cite{KhuriWeinsteinYamada1}. The topology of the domain of outer communication is classified in \cite{KhuriMatsumotoWeinsteinYamada}.
\end{remark}

The reduction to a harmonic map problem in $(i)$ is known (see e.g.  \cite{Bouchareb:2007ax}), although here we provide a simplified proof tailored to the existence problem. The uniqueness result of this theorem is stated for the harmonic map, and while this is an important component of the uniqueness argument for minimal supergravity solutions it is not sufficient on its own to yield this desired conclusion. Corollaries \ref{solitoncorollary} and \ref{generalcorollary} below will address the spacetime uniqueness question, and will identify the parameters needed for a classification.
The existence result in $(ii)$ generalizes that of \cite{KhuriWeinsteinYamada} in the 5D vacuum case. A primary difference with the vacuum theory is that here the potentials are nonconstant along the axes, which results from the inclusion of dipole charge and electric charge. This implies that the construction of an approximate solution, one of the main steps in the proof of existence, is much more elaborate. In addition, it is not obvious in the supergravity setting how the prescription of charges leads to knowledge of the potential constants that determine the rod data set. This leads to complications with the uniqueness question, and is one of the reasons such questions have remained open.

The supergravity solutions produced in $(iii)$ may possess conical singularities. It is conjectured that this is the only obstruction to regularity of the spacetime arising from the harmonic map. In fact, there are two possible regularity issues that should be addressed, namely, the questions of \textit{geometric regularity} and \textit{analytic regularity}. Geometric regularity concerns the ability to smoothly extend the spacetime metric across the orbit space rods, and is directly related to the
potential presence of conical singularities. On the other hand, analytic regularity concerns the differentiability properties of the harmonic map up to the orbit space boundary after the singular part of the map has been removed.  Although solutions to singular harmonic map problems into nonpositively curved targets, such as the one studied here, are expected to be analytically regular, we will not address this topic here. It should be noted that analytic regularity has only relatively recently been established for the Einstein-Maxwell equations in the classical 4D setting by Nguyen \cite{N}, for the interior of axis rods; in the same setting, analytic regularity at the poles has apparently not yet been addressed. On the other hand, the 4D vacuum case was treated independently by Li-Tian \cite{LT,LT93} and Weinstein \cite{Wei92}. In Appendix \ref{sec9} we will show that, assuming a minimal amount of analytic regularity for the harmonic maps constructed above, the question of geometric regularity is resolved precisely when the axes are devoid of conical singularities. 

In the unbalanced case where conical singularities have not been resolved, we may count the number of parameters on which the solution depends. Suppose that the topology is fixed, that is, the rod structures have been chosen. Let $n$ denote the number of axis rods, and $m$ denote the number of horizon rods. For each finite rod its length counts as one parameter, and thus there are $(n-2)+m$ total length parameters. Each horizon component has an electric charge and two angular momenta, which yields $3m$ parameters. Furthermore, it will be shown that we may prescribe $n-2$ dipole charges, and therefore all together the solutions produced are determined by $2(n-2)+4m$ parameters. It is expected that one parameter will be needed to alleviate the conical singularity on each finite axis rod, while the two semi-infinite rods are automatically free of such singularities \cite{KhuriWeinsteinYamada1}; it should be noted that \cite{KhuriWeinsteinYamada1} details a proof for the vacuum case, however, the argument carries over with mild modifications to the current setting. Hence, a balanced solution should be determined by $n-2+4m$ parameters. In particular, solitons (which by definition do not have horizons) are determined by $n-2$ parameters corresponding to the dipole charge of each finite rod. Note that such finite rods represent nontrivial 2-cycles in the spacetime and are referred to as `bubbles'. These solutions do not appear in the vacuum setting, and are of independent interest. Theorem \ref{primaryresult} implies the following result.

\begin{cor}\label{solitoncorollary}
Associated with each soliton solution is a set of $n$ rod structures satisfying the admissibility condition, and $n-2$ dipole charges. Conversely, given an admissible set of $n$ axis rod structures satisfying the compatibility condition, and $n-2$ dipole charges, there exists a unique soliton (possibly with conical singularities on the axes) realizing this data.
\end{cor}

The question of whether the solitons produced by this corollary can be balanced to cure all conical singularities, remains open. On the other hand, the result states that regardless of whether the solitons are balanced or not, they are uniquely determined by their dipole charges and rod structure.

For solitons the total angular momenta $\mathcal{J}_i$ and electric charge $\mathcal{Q}$ are functions of the dipole charges, and therefore the solution may be interpreted as being supported from these fluxes alone. It should also be pointed out that alternate definitions of these charges $J_i$, $Q$ are used in the literature, see for example \cite{Gunasekaran:2016nep}. In Section \ref{sec5} we show how the two definitions are related.

In the general case when horizons are present, Theorem \ref{primaryresult} can be
used to obtain solutions with prescribed angular momenta and electric charge of each horizon component, as well as prescribed dipole charges for bubbles.
In particular, the next result translates the rod data sets of Theorem \ref{primaryresult} into the language of physical charges.

\begin{cor}\label{generalcorollary}
Associated with each well-behaved stationary bi-axisymmetric solution is
a set of $n$ axis rod structures, $m$ horizon rod structures satisfying the admissibility condition, $2m$ angular momenta $\mathcal{J}_i$, $m$ electric charges $\mathcal{Q}$, and $n-2$ dipole charges $\mathcal{D}_l$. Conversely, given an admissible set of $n$ axis and $m$ horizon rod structures satisfying the compatibility condition, along with $2m$ angular momenta, $m$ electric charges, and $n-2$ dipole charges, there exists a unique well-behaved solution (possibly with conical singularities on the axes) of the 5D minimal supergravity equations realizing this data.
\end{cor}

While there are more than $n-2$ dipole charges that one may compute from a sequence
of $n+m$ rod structures, the proof of this corollary gives an algorithm for identifying those $n-2$ dipole charges which may be used to uniquely determine the solution. These solutions, again, may have conical singularities on the axes unless
$n-2$ of the parameters are balanced or chosen appropriately.

The uniqueness of solutions of minimal supergravity was previously considered for particular rod structures~\cite{Tomizawa:2010xj,Tomizawa:2009ua,Tomizawa:2009tb} corresponding to black hole solutions having trivial topology in the domain of outer communication. More recently, the case of general rod structures was studied in \cite{Armas:2014gga,Armas:2009dd} using a Mazur identity appropriate in this setting. In particular, in \cite{Armas:2014gga} it was shown that for a spacetime containing a single non-extreme black hole, uniqueness may be obtained by fixing the rod structure, mass, angular momenta, electric charge, and magnetic flux on each spatial rod. The case of multiple black holes is also considered and arguments supporting a uniqueness result are given under certain special hypotheses.
In addition, these results require knowledge of a higher number of parameters than is necessary to uniquely determine a solution. In contrast, our approach identifies the minimal number of parameters that are needed to uniquely specify a solution, and furthermore shows that for each admissible choice of this set of parameters there is a corresponding solution.

The results above may be generalized by omitting the admissibility condition, in which case the generalized compatibility condition \eqref{compat1} should be imposed. This extra technical condition arises from the particular approach used here to construct the harmonic map. It is not known whether this condition is necessary for existence. When the admissibility condition does not hold, the resulting spacetime will have orbifold singularities associated with the corners. This means that neighborhoods of such points in a time slice are foliated by lens spaces instead of spheres.

\begin{theorem}\label{primaryresult2}
Given a rod data set respecting the generalized compatibility condition, there exists a unique harmonic map $\varphi\colon\mathbb{R}^3\setminus\Gamma\rightarrow G_{2(2)}/SO(4)$ having the prescribed singularities and potential values on $\Gamma$ associated with the rod data. This map produces
a well-behaved stationary bi-axially symmetric solution of the 5D minimal supergravity equations without degenerate horizons.
\end{theorem}

The organization of this paper is as follows. In Sections \ref{sec3} and \ref{sec4} dimensional reduction is carried out to the 2D orbit space, and the harmonic map
problem is defined. Section \ref{sec5} is dedicated to a description and relation
between the various charges associated with stationary bi-axisymmetric minimal supergravity. An approximate solution to the singular harmonic map problem is constructed in Section \ref{sec6}, and the full existence/uniqueness is carried out together with the proofs of our main results in Sections \ref{sec7} and \ref{sec8}. Lastly, two appendices are included to address the relationship between conical singularities and geometric regularity, and to compute the Mazur quantity.

\section{Minimal Supergravity and Reduction to a 3D Wave Map}
\label{sec3}\setcounter{equation}{0}
\setcounter{section}{3}

\subsection{Field Equations}

We will consider five dimensional spacetimes $(M^5, \mathbf{g},\mathcal{F})$ where $M^5$ is a smooth, orientable manifold equipped with a Lorentzian metric $\mathbf{g}$ having signature $(-,+,+,+,+)$ and $\mathcal{F}$ is a closed 2-form describing the Maxwell field.  A solution $(M^5, \mathbf{g},\mathcal{F})$ of $D=5$ minimal supergravity is a critical point of the following action functional
\begin{equation}\label{action}
\mathcal{S} = \int_{M^5} R \star 1 - \frac{1}{2} \mathcal{F} \wedge \star \mathcal{F} - \frac{1}{3\sqrt{3}}  \mathcal{F} \wedge \mathcal{F} \wedge \mathcal{A},
\end{equation}
where $R$ is scalar curvature, $\star$ is the Hodge dual operator associated to $\mathbf{g}$, and $\mathcal{A}$ is a local 1-form gauge potential $\mathcal{F}= d\mathcal{A}$.  In general $H_2(M^5) \neq 0$ so $\mathcal{A}$ need not be globally defined.  This theory automatically includes vacuum general relativity when $\mathcal{F} \equiv 0$.  The spacetime field equations derived from this functional are
\begin{align} \label{SFEintro}
\begin{split}
&R_{ab} = \frac{1}{2} \mathcal{F}_{ac} \mathcal{F}_b^{~c} - \frac{1}{12} |\mathcal{F}|_{\mathbf{g}}^2 \mathbf{g}_{ab} , \\
&d\star \mathcal{F} + \frac{1}{\sqrt{3}} \mathcal{F} \wedge \mathcal{F} =0 .
\end{split}
\end{align}
Unlike the more familiar pure Einstein-Maxwell system, $d \star \mathcal{F} \neq 0$. In what follows dimensional reduction of the minimal supergravity equations will be carried out, leading to a sigma model or harmonic map system. Although such a reduction has previously been given in the literature, the methods of our existence and uniqueness result require a particular formulation which we now describe. Moreover, simplified proofs of some known identities are also given.

\subsection{Bi-Axisymmetric Spacetimes}\label{biaxisymmetry}

Suppose that the spacetime $(M^5,\mathbf{g})$ admits an action of $U(1)\times U(1) \equiv U(1)^2$ by isometries. Denote the generators of this action by $\eta_{(i)}$, $i = 1,2$ and their positive semidefinite matrix of inner products by $f_{ij} = \mathbf{g}(\eta_{(i)}, \eta_{(j)})$ with $f=\det f_{ij}$.  The $\eta_{(i)}$ are Killing vector fields, that is $\mathcal{L}_{\eta_{(i)}} \mathbf{g} =0$, they commute with each other $\mathcal{L}_{\eta_{(i)}} \eta_{(j)} =0$, and the conserved angular momenta associated with this symmetry may be encoded in the
\emph{twist 1-forms}
\begin{equation}\label{Thetadefinition}
\Theta_i=\star \left(\eta_{(1)}\wedge\eta_{(2)}\wedge d\eta_{(i)}\right).
\end{equation}
Geometrically these forms measure the integrability of the 3-plane distribution orthogonal to the $U(1)^2$ action. We will
denote the natural inner product on forms by $(\cdot,\cdot)$, the interior product operator by $\iota$, and the wave operator acting on functions by $\Box_{\mathbf{g}}=\nabla^2_{\mathbf{g}} = -\star d\star d$.

\begin{proposition}\label{KVFiden}
The following identities for the twist 1-forms and fiber metric hold:
\begin{eqnarray}
	d \Theta_i &=& -2\iota_{\eta_{(1)}}\iota_{\eta_{(2)}} \star \mathrm{Ric}(\eta_{(i)}), \label{iden1} \\
	\star d \star \left(f^{-1}f^{ij}\Theta_i\right) &=& 0, \label{iden2} \\
	\Box_{\mathbf{g}} f_{ij}  &=& f^{lm} ( df_{li} , df_{mj}) - f^{-1} (\Theta_i , \Theta_j)  - 2\mathrm{Ric}(\eta_{(i)},\eta_{(j)}). \label{iden3}
\end{eqnarray}
\end{proposition}

\begin{proof}
The first equation \eqref{iden1} arises from the identity $\star d \star d K = -2 \mathrm{Ric}(K)$ for Killing fields $K$, Cartan's formula, as well as the identities\footnote{See \cite[Appendix D]{AKK1} for our conventions for forms and relevant formulas.}
\begin{equation}\label{iden4}
\iota_X \star \alpha = (-1)^p \star (X \wedge \alpha), \qquad \iota_X  \alpha = (-1)^{p} \star (X \wedge \star \alpha),
\end{equation}
for $p$-forms $\alpha$ and vector fields $X$ where without ambiguity the dual 1-form to $X$ is denoted by the same notation.
Next, observe that a direct calculation gives the exterior derivative of the dual 1-forms to the Killing field generators
\begin{equation}\label{deta}
d \eta_{(i)} = -f^{-1}\star (\eta_{(1)} \wedge \eta_{(2)} \wedge \Theta_i) + f^{kj}  d f_{ji} \wedge \eta_{(k)}.
\end{equation}
This may be rewritten as
\begin{equation}
d [ f^{mi} \eta_{(i)}] = -\star (\eta_{(1)} \wedge \eta_{(2)} \wedge \mu^m)
\end{equation}
where $\mu^i = f^{-1} f^{ij} \Theta_j$, and therefore with the help of Cartan's formula
\begin{equation}\label{iden5}
0 = d\star (\eta_{(1)} \wedge \eta_{(2)} \wedge \mu^m) = -\iota_{\eta_{(1)}} d \star (\eta_{(2)} \wedge \mu^m) = -\star(\eta_{(1)} \wedge \star d \star (\eta_{(2)} \wedge \mu^m)).
\end{equation}
Consider now the following identity, which holds for an arbitrary Killing field $K$ and applies to $p$-forms
\begin{equation}
\mathcal{L}_K \alpha = (-1)^{p} ( K \wedge \star d \star \alpha) + (-1)^{p+1}\star d \star (K \wedge \alpha).
\end{equation}
Using this together with \eqref{iden5} and the fact that $\mathcal{L}_{\eta_{(2)}} \mu^m =0$ produces
\begin{equation}
(\eta_{(1)} \wedge \eta_{(2)} )\star d \star \mu^m = \eta_{(1)} \wedge \star d \star (\eta_{(2)} \wedge \mu^m) = 0.
\end{equation}
Since $\eta_1 \wedge \eta_2 \neq 0$ we conclude
\begin{equation}
\star d \star \left(f^{-1}f^{ij}\Theta_j\right) =0,
\end{equation}
or equivalently that $\mu^m$ has vanishing spacetime divergence. This establishes, independently of any field equations, formula \eqref{iden2}.

Lastly, taking the inner product of \eqref{deta} with itself yields
\begin{equation}
(d \eta_{(i)} , d \eta_{(j)}) = f^{ml}(df_{mi},d f_{lj}) - f^{-1} (\Theta_i , \Theta_j).
\end{equation}
Moreover since
\begin{equation}
d f_{ij} = d \left[i_{\eta_{(i)}} \eta_{(j)}\right]  = -\star (\eta_{(i)} \wedge \star d\eta_{(j)}),
\end{equation}
we obtain
\begin{equation}
d\star d f_{ij} = d\eta_{(i)} \wedge \star d \eta_{(j)} - \eta_{(i)} \wedge d \star d \eta_{(j)} = \left(( d\eta_{(i)} , d \eta_{(j)})  + (\eta_{(i)} , \star d \star \eta_{(j)} ) \right) \epsilon
\end{equation}
where $\epsilon$ is the spacetime volume form.  Therefore
\begin{align}
\begin{split}
\Box_{\mathbf{g}} f_{ij} =& (d \eta_{(i)} , d \eta_{(j)}) - 2\text{Ric}(\eta_{(i)},\eta_{(j)}) \\
=& f^{ml} ( df_{mi} , df_{lj}) - f^{-1} (\Theta_i , \Theta_j)  - 2\text{Ric}(\eta_{(i)},\eta_{(j)}),
\end{split}
\end{align}
which establishes \eqref{iden3}.
\end{proof}

In order to elucidate the wave map structure underlying the field equations, we must reduce the analysis to the 3-dimensional Lorentzian orbit space $\hat{M}^3=M^5/U(1)^2$. Since this Kaluza-Klein decomposition procedure is well understood \cite{Maison:1979kx}, only the main features will be mentioned.  The spacetime metric may be decomposed as
\begin{equation}
\mathbf{g}_{ab} =f^{-1}h_{ab}+{f}^{ij}\eta_{(i)a}\eta_{(j)b}
=f^{-1}h_{ab}+\Phi^T_a F^{-1}\Phi_b,
\end{equation}
where $f^{-1}h_{ab}$ is the orbit space metric, $\Phi=(\eta_{(1)},\eta_{(2)})^T$, and $F=(f_{ij})$. Next let
\begin{equation}
h^a_{~b}=\delta^a_{~b}-{f}^{ij}\eta_{(i)}^a\eta_{(j)b}
\end{equation}
be the projection tensor onto $\hat{M}^3$,
then a detailed computation reveals the relation between the Ricci tensors of the spacetime and orbit space. Namely, if $\Omega=(\Theta_1,\Theta_2)^T$ then
\begin{align}
\begin{split} \label{Ric(h)}
\mathrm{Ric}(h)_{ac}=&\left[{h}_{a}^b{h}_{c}^d
+f^{-1}{h}_{ac}(\Phi^T)^b{F}^{-1}\Phi^d\right] \mathrm{Ric}(\mathbf{g})_{bd}+\frac{1}{2f} \Omega_a^T{F}^{-1}\Omega_c\\
&
+\frac{1}{4}\text{Tr}\left({F}^{-1}\partial_a{F}{F}^{-1}\partial_c{F}\right)
+\frac{\partial_a f\partial_c f}{4f^2}.
\end{split}
\end{align}
Furthermore the differential identities of Proposition \ref{KVFiden} may be rewritten as
\begin{align}
\begin{split}\label{diffidentities}
\hat{\nabla}_a\Omega_b-\hat{\nabla}_b\Omega_a =&2\text{Ric}(\mathbf{g})^{c}{}_{d}\Phi^d
\eta_{(1)}^{m}\eta_{(2)}^{l}\epsilon_{abcml},\\
\operatorname{div}_h \left( f^{-1}F^{-1} \Omega\right) =& h^{ab}\hat{\nabla}_a\left(f^{-1}{F}^{-1}\Omega_b\right)=0,\\
	\Box_{h}F=& h^{ab}\hat{\nabla}_a F \left(F^{-1}\partial_b F\right)-f^{-1}h^{ab}\Omega_a\Omega^T_b-2f^{-1}\text{Ric}(\mathbf{g})_{ad}\Phi^a(\Phi^T)^d,
\end{split}
\end{align}
where $\hat{\nabla}$ is the connection associated to $h$.

\subsection{The Potentials}

The wave map of minimal supergravity is constructed in part from five potentials \cite{Bouchareb:2007ax,Kunduri:2011zr}. In particular
two twist potentials $\zeta_1$, $\zeta_2$ encode angular momentum, two magnetic potentials $\psi_1$, $\psi_2$ encode the
dipole charges, and one electric potential $\chi$ is associated with the electric charge. The global definition of these potentials is guaranteed as the orbit space $\hat{M}^3$ is simply connected \cite{Hollands:2012xy}. Observe first that Cartan's formula and the fact that $d\mathcal{F}=0$ imply the existence of magnetic potentials
\begin{equation}\label{psi^i}
d \psi_i = \iota_{\eta_{(i)}} \mathcal{F}.
\end{equation}
It is straightforward to show that
\begin{equation}
\mathcal{L}_{\eta_{(i)}}\psi_j = \iota_{\eta_{(i)}} \iota_{\eta_{(j)}} \mathcal{F} =0,
\end{equation}
so that the magnetic potentials are functions defined on the orbit space.

Consider now the 1-form
\begin{equation}\label{Upsilon}
\Upsilon = -\iota_{\eta_{(1)}} \iota_{\eta_{(2)}} \star \mathcal{F}.
\end{equation}
As a consequence of the Maxwell equation in \eqref{SFEintro} we have
\begin{equation}
d\Upsilon =\frac{1}{\sqrt{3}} d\left(\psi_1 d\psi_2 - \psi_2 d\psi_1 \right),
\end{equation}
which yields existence of an electric potential satisfying
\begin{equation}\label{uuuu}
d\chi = \Upsilon - \frac{1}{\sqrt{3}}\left(\psi_1 d\psi_2 - \psi_2 d\psi_1 \right).
\end{equation}
Next, recall that in pure vacuum the twist 1-forms $\Theta_i$ are closed.  In the supergravity setting this is no longer case since the Ricci tensor is nonvanishing. Using the field equations \eqref{SFEintro}, a detailed calculation \cite{Kunduri:2011zr} shows that
\begin{equation}
d\Theta_i=-\Upsilon\wedge \iota_{\eta_{(i)}} \mathcal{F}
= d\left[\psi_i \left(d\chi + \frac{1}{3\sqrt{3}}(\psi_1 d\psi_2 - \psi_2 d\psi_1)\right)\right].
\end{equation}
It follows that there exist twist potentials which obey
\begin{equation}\label{eq2.11}
d\zeta_i = \Theta_i - \psi_i \left[d\chi + \frac{1}{3\sqrt{3}}(\psi_1 d\psi_2 - \psi_2 d\psi_1)\right],
\end{equation}
and it is routine to show that $\zeta_i$ as well as $\chi$ are functions on
the orbit space. Finally, note that the Maxwell field can be reconstructed from the fields $(f_{ij}, \zeta_i,\chi, \psi_i)$ with the identity
\begin{equation}\label{Maxwell}
\mathcal{F} = f^{-1}\star(\eta_{(2)} \wedge \eta_{(1)} \wedge \Upsilon) + {f}^{ij} \eta_{(i)} \wedge d \psi_j .
\end{equation}

\begin{proposition}\label{SFEreduced}
The supergravity field equations \eqref{SFEintro} for $U(1)^2$-invariant solutions $(\mathbf{g},\mathcal{F})$ are equivalent to the following system:
\begin{align}
\begin{split}\label{reducedEFE}
\mathrm{Ric}(h)_{ab} =&  \frac{1}{4}\mathrm{Tr}\left(F^{-1}\partial_a F F^{-1}\partial_b F\right)+\frac{\partial_a f\partial_b f}{4f^2}\\
&+ \frac{\Upsilon_a \Upsilon_b}{2   f} + \frac{f^{ij}}{2} \partial_a \psi_{i} \partial_b \psi_{j} + \frac{1}{2 f}f^{ij} \Theta_{ia} \Theta_{jb},
\end{split}
\end{align}
\begin{equation}\label{laplacianlambda}
\Box_h {f}_{ij} = {f}^{kl}(d f_{ik} , df_{jl})_h - f^{-1}(\Theta_i ,  \Theta_j)_h - (d\psi_i ,d\psi_j)_h +\frac{1}{3} f_{ij}  \left(f^{kl} (d\psi_k , d\psi_l)_h - f^{-1}(\Upsilon,\Upsilon)_h\right),
\end{equation}
where $(\cdot, \cdot)_h$ denotes the inner product on forms with respect to the metric $h$ and
\begin{equation}\label{lapzeta}
\mathrm{div}_h \left(f^{-1}f^{ij} \Theta_j\right) =0 ,
\end{equation}
\begin{equation}\label{lappsi}
\mathrm{div}_h (f^{ij} d\psi_j) = (\Upsilon, f^{-1}f^{ij}\Theta_j)_h - \frac{2}{\sqrt{3}f} \left[ \delta^i_{~2} (\Upsilon,d\psi_1)_h - \delta^i_{~1} (\Upsilon,d\psi_2)_h \right] ,
\end{equation}
\begin{equation}\label{lapchi}
\mathrm{div}_h \left(f^{-1}\Upsilon \right) = -(d\psi_i, f^{-1}f^{ij} \Theta_j)_h .
\end{equation}
\end{proposition}

\begin{proof}
To obtain \eqref{reducedEFE} one performs a long computation using the expression for the Maxwell field \eqref{Maxwell} to evaluate the spacetime Ricci tensor, and then substitutes the result into \eqref{Ric(h)}.  The field equation \eqref{laplacianlambda} is similarly obtained using \eqref{diffidentities}, and Proposition \ref{KVFiden} gives \eqref{lapzeta} directly.  To derive \eqref{lappsi}, use \eqref{iden4} to find
\begin{equation}
\star d \star (f^{ij} d\psi_j) = \star d \star f^{ij} \star (\eta_{(j)} \wedge \star \mathcal{F}) = \star \left(f^{ij}\eta_{(j)} \wedge d \star \mathcal{F} -d( f^{ij} \eta_{(j)}) \wedge \star \mathcal{F} \right).
\end{equation}
Now employ the identity \eqref{deta} and the Maxwell equation \eqref{SFEintro} on the first and second terms respectively, as well as \eqref{Maxwell}, to obtain the desired equation. With the help of $\Upsilon = f\star (f^{1i}\eta_{(i)} \wedge f^{2j} \eta_{(j)} \wedge \mathcal{F})$, a similar computation leads to \eqref{lapchi}.
\end{proof}
	
The equations of Proposition \ref{SFEreduced} arise as
critical points of a 3-dimensional theory of gravity on $(\hat{M}^3,h)$ coupled to a wave map \cite{Bouchareb:2007ax,Possel:2003rk} (that is, a harmonic map defined on a manifold with Lorentzian signature) having noncompact symmetric space target manifold $G_{2(2)}/SO(4)$ and governed by the action
\begin{equation}\label{wavemap}
\mathcal{S}[h, X] = \int_{\hat{M}^3} \left(R_h - 2h^{\mu \nu} G_{AB} \partial_\mu X^A \partial_\nu X^B \right) \; d \mathrm{Vol}(h)
\end{equation}
where $R_h$ is the scalar curvature of $h$, and $X = (f_{ij},\zeta_i, \chi, \psi_i)$ are the set of coordinates on the target manifold with metric
\begin{align}\label{yyyy}
\begin{split}
8G_{AB} dX^A dX^B =&(d \log f)^2 + \mathrm{Tr} ( F^{-1} d F)^2 + 2f^{-1}f^{ij} \Theta_i \Theta_j + 2f^{-1}\Upsilon^2 + 2 f^{ij} d\psi_i d\psi_j .
\end{split}
\end{align}
The Euler-Lagrange equations of \eqref{wavemap} are given by
\begin{align}\label{17ji}
\begin{split}
\mathrm{Ric}(h)_{\mu \nu} =&\frac{1}{8} {\text{Tr}} (\Psi^{-1} \partial_\mu \Psi\Psi^{-1} \partial_\nu \Psi) = 2G_{AB}\partial_\mu X^A \partial_\nu X^B, \\
\hat{\nabla}^\mu (\Psi^{-1} \partial_\mu \Psi) =&0 .
\end{split}
\end{align}
An explicit expression for the positive definite unimodular coset representative $\Psi$ may be found in \cite{Kunduri:2011zr}.

\subsection{Construction of the Solution From Potentials}
\label{yyyyy}

Given a solution $(h, f_{ij},\zeta_i, \chi, \psi_i)$ of the reduced supergravity equations, one may reconstruct the full spacetime solution $(\mathbf{g},\mathcal{F})$. To see this introduce local coordinates $x^\mu$ on the orbit space $(\hat{M}^3,h)$, and $2\pi$-periodic coordinates $\phi^i$ adapted to the Killing vectors so that $\eta_{(i)}=\partial_{\phi^i}$. The dual 1-forms $\mathbf{g}(\eta_{(i)}, \cdot)$ to the Killing fields are then given by
${f}_{ij}(d\phi^j + A^j)$, where $A^j = A^j_\mu dx^\mu$ are 1-forms on the orbit space. The spacetime metric then takes on the following expression
\begin{equation}\label{3dreduction}
\mathbf{g} = f^{-1}h_{\mu \nu} dx^\mu \otimes dx^\nu + f_{ij}(d\phi^i + A^i)\otimes(d \phi^j + A^j) .
\end{equation}
A simple calculation shows that the $A^j$ are determined by the twist 1-forms via
\begin{equation}\label{dA}
dA^i = -\star_h \left(f^{-1}f^{ij}\Theta_j\right),
\end{equation}
where $\star_h$ denotes the Hodge dual with respect to the metric on the orbit space.  Observe that integrability of this equation is guaranteed by the second equation of \eqref{diffidentities}.

To construct the Maxwell field, first note that $\Upsilon$ is a 1-form defined on the base space, that is it may be expressed as $\Upsilon = \Upsilon_\mu dx^\mu$.  We may then compute
\begin{equation}
\star (\eta_{(2)} \wedge \eta_{(1)} \wedge \Upsilon) = -f^{-1} \star_h \Upsilon.
\end{equation}
Furthermore
\begin{equation}
f^{ij} \eta_{(i)} \wedge d \psi_j  = d \phi^i \wedge d\psi_i + A^j \wedge d\psi_j = -d \left[\psi_j (d\phi^j + A^j)\right] + \psi_j dA^j .
\end{equation}
It now follows from \eqref{Maxwell} that
\begin{equation}\label{F}
\mathcal{F} = f^{-2}\star_h \left[ f   \psi_j f^{jk}\Theta_k  - \left(d \chi + \frac{1}{\sqrt{3}}\left(\psi_1 d\psi_2 - \psi_2 d\psi_1\right)\right)\right] -d \left[\psi_j (d \phi^j + A^j) \right],
\end{equation}
thus completing the construction of the spacetime solution.

\section{Reduction to 2D}
\label{sec4}\setcounter{equation}{0}
\setcounter{section}{4}

In this section we will assume that in addition to being bi-axisymmetric the spacetime is also stationary, that is the group of isometries is $U(1)^2\times \mathbb{R}$. Thus, along with the rotational Killing field generators $\eta_{(i)}$ there is another Killing field $\xi$ which asymptotically coincides with the generator of time translations at spatial infinity and which commutes with $\eta_{(i)}$; we have
$\mathcal{L}_\xi \mathbf{g} = 0$ and $\mathcal{L}_\xi \mathcal{F} =0$.
Scalar potentials associated to the Maxwell field may be introduced with the help of this new Killing field. These will again be globally defined due to simple connectedness \cite{Hollands:2007aj} of the 2-dimensional orbit space $\hat{M}^2=M^{5}/[U(1)^2 \times\mathbb{R}]$. In particular there is an `electric' potential satisfying
\begin{equation}
d\mathcal{E}_0 = \iota_\xi \mathcal{F}.
\end{equation}
Furthermore the following 2-form is closed
\begin{equation}
\Xi \equiv \frac{1}{2} \iota_\xi \star \mathcal{F} - \frac{1}{\sqrt{3}} \mathcal{E}_0 \mathcal{F},
\end{equation}
which implies the existence of two more potentials
\begin{equation}
d\mathcal{E}_i = \iota_{\eta_{(i)}} \Xi = \frac{1}{2} \iota_{\eta_{(i)}} \iota_\xi \star \mathcal{F} - \frac{1}{\sqrt{3}} \mathcal{E}_0 d\psi_i,\quad\quad i=1,2.
\end{equation}

The 2-plane distribution orthogonal to the three symmetry generators is integrable by Frobenius' theorem.  This requires
\begin{equation}
\star(\xi \wedge \eta_{(1)} \wedge \eta_{(2)} \wedge d K_I) =0
\end{equation}
for each $I=0,1,2$, where the $K_I$ are used to denote the three Killing fields and their duals.
To see that this holds for $I=1,2$ observe that
\begin{equation}
\star(\xi \wedge \eta_{(1)} \wedge \eta_{(2)} \wedge d \eta_{(i)}) = \iota_{\xi} \Theta_i =0,
\end{equation}
where we have used \eqref{eq2.11} and the fact that all scalar potentials are invariant under $\xi$. For example, note that
\begin{equation}
d \iota_\xi d\psi_i =d ( \iota_\xi \iota_{\eta_{(i)}} \mathcal{F}) =0
\end{equation}
as $d\mathcal{F} =0$, and so $\iota_\xi d\psi_i$ is constant. If we choose $\eta_{(i)}$ to be one of the generators of an axis of symmetry at spatial infinity then this constant must vanish, and therefore $\iota_\xi d\psi_i=0$.  Analogous arguments show that the other scalar potentials are also invariant under $\xi$.
Consider now the case when $I=0$ and compute
\begin{equation}\label{oooo}
d \star(\xi \wedge \eta_{(1)} \wedge \eta_{(2)} \wedge d \xi ) = 2\iota_\xi \iota_{\eta_{(1)}} \iota_{\eta_{(2)}} \star \text{Ric}(\mathbf{g}) (\xi).
\end{equation}
Furthermore from the field equations it can be shown that
\begin{equation}
\star \text{Ric}(\mathbf{g}) (\xi) = -\frac{1}{3} \Xi \wedge \mathcal{F} + \frac{1}{3} d \left(\mathcal{E}_0 \star \mathcal{F}\right),
\end{equation}
and expressing this in terms of the various potentials implies that \eqref{oooo} vanishes. Thus $\star(\xi \wedge \eta_{(1)} \wedge \eta_{(2)} \wedge d \xi)$ is a constant, which must vanish since at least two linear combinations of the $\eta_{(i)}$ vanish along the symmetry axes at spatial infinity.
Hence Frobenius' theorem applies.

Introduce now a time coordinate $t$ such that $\xi = \partial_t$. The orthogonal transitivity of the isometry group allows for the following expression of the spacetime metric
\begin{equation}
 \mathbf{g}=f^{-1}g_2 -f^{-1}\rho^2 dt^2 + f_{ij} (d\phi^i + \omega^i dt)(d\phi^j + \omega^j dt),
\end{equation}
so that $h = g_2 -\rho^2 dt^2$ and $A^i = \omega^i dt$ in \eqref{3dreduction}. Here $g_2$ is the $\hat{M}^2$ orbit space metric induced by $h$, and $\rho = \sqrt{-\det q}$ where $q_{IJ} = \mathbf{g}(K_I, K_J)$ is the fibre metric obtained by restricting $\mathbf{g}$ to the Killing fields, that is
\begin{equation}
q = -f^{-1}\rho^2 dt^2 + f_{ij} (d\phi^i + \omega^i dt)(d\phi^j + \omega^j dt).
\end{equation}
It is proved in \cite[Theorem 5.1]{Chrusciel} that $\det q <0$ in the domain of outer communications.
This simplified expression for $h$ yields
\begin{equation}
\text{Ric}(h)_{tt} = \rho \Delta_2 \rho, \qquad \text{Ric}(h)_{ab} = \text{Ric}(g_2)_{ab} - \frac{1}{\rho} D_a D_b \rho,
\end{equation}
where $\Delta_2$, $D_a$ are the Laplacian and covariant derivative associated with $g_2$. From \eqref{reducedEFE} and the fact that all quantities are independent of $t$, it follows that $\Delta_2 \rho =0$ so that as in the vacuum case $\rho$ is harmonic with respect to $g_2$. From this it can be shown \cite[\S 6]{ChruscielCosta} that $\rho$ has no critical points in the orbit space $\hat{M}^2$. We may then define the harmonic conjugate function $z$ up to a constant by $dz = \star_2 d \rho$. The functions $(\rho,z)$ form a global set of coordinates on the orbit space which is homeomorphic to the right-half plane $\{(\rho,z)\mid \rho>0\}$. These coordinates are also naturally
isothermal so that there is a function $\sigma$ defined on the orbit space such that
\begin{equation}
g_2 = e^{2\sigma} (d\rho^2 + dz^2).
\end{equation}

Concerning the Maxwell field, it also simplifies considerably in Weyl-Papapetrou coordinates. According to \eqref{dA} we have
\begin{equation}\label{fjhkjhhw}
d\omega^i = \rho f^{-1}f^{ik} \star_2\Theta_k,
\end{equation}
and therefore \eqref{F} becomes
\begin{align}
\begin{split}
\mathcal{F} =& dt \wedge \left[ d\mathcal{E}_0 - d (\psi_j \omega^j)\right] - d\left[\psi_j(d\phi^j + \omega^j dt)\right] \\
& = -d \left[ (\mathcal{E}_0 - \psi_j \omega^j) dt + \psi_j (d\phi^j + \omega^j dt)\right] \\
&=-d \left[ \mathcal{E}_0 dt + \psi_j d\phi^j\right]
\end{split}
\end{align}
with the help of
\begin{equation}
d\mathcal{E}_0 = -\rho f^{-2}\star_2 \left[d\chi + \frac{1}{\sqrt{3}}(\psi_1 d\psi_2 - \psi_2 d\psi_1)\right] + \psi_j d\omega^j + d(\psi_j \omega^j).
\end{equation}
In addition it should be pointed out that a useful advantage of these coordinates is that the $h$-Laplacian of any function $u$ defined on the orbit
space becomes
\begin{equation}
\Delta_h u = \frac{1}{\sqrt{\det h}} \partial_{a} (\sqrt{\det h}h^{ab} \partial_b u) = e^{-2\sigma} \Delta u,
\end{equation}
where $\Delta$ is the Laplacian for an auxiliary Euclidean 3-space in which the flat metric is written in cylindrical coordinates
\begin{equation}
\delta= d\rho^2 + dz^2 + \rho^2 d\phi^2.
\end{equation}
Here $\phi$ is an auxiliary azimuthal angle on which no quantity depends. Therefore, in the harmonic map system described below it is this flat Laplacian that appears.

It will now be shown that the only content of the 3D Einstein equations of the system \eqref{17ji} is to determine $\sigma$ via quadrature.
Observe that
\begin{equation}
\text{Ric}(h)_{ab} = -\delta_{ab} \Delta_2 \sigma - \rho^{-1} D_a D_b \rho,
\end{equation}
and therefore
\begin{equation}
\text{Ric}(h)_{\rho \rho} = -\Delta_2 \sigma + \rho^{-1} \partial_\rho \sigma, \quad \text{Ric}(h)_{\rho z} = \rho^{-1} \partial_z \sigma, \quad \text{Ric}(h)_{z z} = -\Delta_2 \sigma - \rho^{-1} \partial_\rho \sigma.
\end{equation}
It now follows from \eqref{17ji} that
\begin{align}\label{PDEsigma}
\begin{split}
\rho^{-1} \partial_z \sigma =& 2 G_{AB} \partial_\rho X^A \partial_z X^B, \\
 \rho^{-1} \partial_\rho \sigma =&   G_{AB} \partial_\rho X^A \partial_\rho X^B -  G_{AB} \partial_z X^A \partial_z X^B.
\end{split}
\end{align}
These first order equations for $\sigma$ are integrable as a result of the harmonic map equations. To see this note that the harmonic map equations arise from the action
\begin{equation}\label{HMapaction}
\mathcal{S}_{X} = \int_{\mathbb{R}^3} G_{AB} d X^A \wedge \star_{\delta} dX^B,
\end{equation}
so that the associated divergence free stress-energy tensor is given by
\begin{equation}
T_{ij} = G_{AB} \partial_i X^A \partial_j X^B - \frac{1}{2} \delta_{ij} |dX|^2_G.
\end{equation}
The equations \eqref{PDEsigma} may now be rewritten as
\begin{equation}
\partial_\rho \sigma = 2 \rho T_{\rho \rho}  = - 2\rho T_{zz}, \qquad \partial_z \sigma = 2 \rho T_{\rho z}.
\end{equation}
Next compute
\begin{equation}
\iota_{\partial_\phi} \star_\delta \iota_{\partial_z}T =  \rho T_{zz} d\rho - \rho T_{\rho z} dz = - \rho T_{\rho \rho} d\rho - \rho T_{\rho z} dz = -\frac{1}{2}d \sigma.
\end{equation}
We then have that the integrability of \eqref{PDEsigma} follows from
\begin{equation}
d \left(\iota_{\partial_\phi} \star_\delta \iota_{\partial_z}T\right)  = [\partial_z (\rho T_{zz})+ \partial_\rho (\rho T_{\rho z})]d\rho \wedge dz =
\rho(\operatorname{div} T)(\partial_z) d\rho \wedge dz= 0,
\end{equation}
where $\operatorname{div}$ is the divergence with respect to $\delta$.

In summary, given data $(f_{ij},\zeta_i,\chi,\psi_i)$ forming the coset representative $\Psi$ satisfying the harmonic map equations
\begin{equation} \label{HMap}
\text{div} (\Psi^{-1} \nabla \Psi) =0
\quad \Leftrightarrow \quad \delta^{ab} \partial_a (\rho \Psi^{-1} \partial_b \Psi) =0,
\end{equation}
a spacetime metric $\mathbf{g}$ and Maxwell field $\mathcal{F}$ may be constrcuted yielding a full solution of \eqref{17ji}. Hence, the stationary bi-axisymmetric supergravity equations reduce to the study of a singular harmonic map problem from $\mathbb{R}^3\setminus \Gamma \rightarrow G_{2(2)}/SO(4)$, where $\Gamma$ represents the axes of rotation in the auxiliary orbit space $\mathbb{R}^3$ where $\Psi$ blows-up.

\section{Angular Momentum and Charges}
\label{sec5}\setcounter{equation}{0}
\setcounter{section}{5}

As described in the previous section a well-behaved stationary bi-axisymmetric solution of the minimal supergravity equations admits a global system of Weyl-Papapetrou coordinates in its domain of outer communication $M^5$, so that the metric and Maxwell field are expressed by
\begin{equation}\label{fkiisjgjakng}
\mathbf{g}=f^{-1}e^{2\sigma}(d\rho^2+dz^2)-f^{-1}\rho^2 dt^2
+f_{ij}(d\phi^{i}+\omega^{i}dt)(d\phi^{j}+\omega^{j}dt),\qquad
\mathcal{F}=-d \left[ \mathcal{E}_0 dt + \psi_j d\phi^j\right].
\end{equation}
The 2-dimensional orbit space $\hat{M}^2=M^{5}/[U(1)^2 \times\mathbb{R}]$ is homeomorphic to the right-half plane $\{(\rho,z)\mid \rho>0\}$, and
its boundary $\rho=0$ encodes all nontrivial topology of the spacetime \cite[Theorem 9]{Hollands:2012xy}. This may be described by the rod data
on the $z$-axis that indicates which 1-cycles in the 2-torus fibers vanish \cite{Hollands:2007aj}. In particular
the $z$-axis is broken into $L+1$ intervals called rods
\begin{equation}\label{rods}
\Gamma_{1}=[z_{1},\infty),\text{ }\Gamma_{2}=[z_2,z_1],\text{ }\ldots,\text{ }
\Gamma_{L}=[z_{L},z_{L-1}],\text{ }\Gamma_{L+1}=(-\infty,z_{L}],
\end{equation}
on which either $F=(f_{ij})$ is full rank and the interval is referred to as a horizon rod, or it fails to be of full rank and the interval is referred to as an axis rod. In the case of an axis rod $\Gamma_l$, the kernel of $F$ is 1-dimensional and there is a pair of relatively prime integers $(v_{l}^1,v_{l}^2)$ such that the Killing field $v_l^i \partial_{\phi^i}$ vanishes on $\Gamma_{l}$ \cite[Prop. 1]{Hollands:2007aj}.   The pair $(v_{l}^1,v_{l}^2)$ is called the rod structure of the rod $\Gamma_{l}$, and $(0,0)$ is reserved for the rod structure of a horizon rod. The possible horizon topologies in this setting
are the sphere $S^3$, ring $S^1\times S^2$, and lens space $L(p,q)=S^3/\mathbb{Z}_p$. These topologies may be obtained from a horizon rod which is bounded by two axis rods having the rod structures $\{(1,0),(0,1)\}$, $\{(1,0),(1,0)\}$, and $\{(1,0),(q,p)\}$ respectively. Similarly, if at infinity the two semi-infinite rods possess these pair of rod structures then the resulting spacetime is asymptotically flat (AF), asymptotically Kaluza-Klein (AKK), and asymptotically locally Euclidean (ALE) respectively.
See \cite[Section 3.1]{Hollands:2012xy} for the relevant definitions concerning the asymptotic conditions.

Two consecutive axis rods are separated by a point referred to as a \textit{corner}. In order to preserve the manifold structure of the spacetime, the two neighboring rod structures $v_l^i$ and $v_{l+1}^i$ associated with a corner must satisfy the
\textit{admissibility condition}
\begin{equation} \label{admiss}
\det\begin{pmatrix} v_l^1 & v_l^2 \\ v_{l+1}^1 & v_{l+1}^2 \end{pmatrix} = \pm 1.
\end{equation}
If this does not hold then the spacetime will have an orbifold singularity \cite[Proposition 1]{Hollands:2007aj}.
In addition to \eqref{admiss}, the existence results of this paper rely on a further condition relating the rod structures referred to as the \textit{compatibility condition}. This, however, is only needed in the presence of three consecutive axis rods. Let $\Gamma_{l-1}$, $\Gamma_l$, and $\Gamma_{l+1}$ be such a configuration with rod structures satisfying the admissibility condition at the two corners. We may assume without loss of generality that the determinant in \eqref{admiss} is $+1$
by multiplying the rod structures by $-1$ is necessary. Then the compatibility condition asserts that
\begin{equation}\label{compat}
v_{l-1}^1 v_{l+1}^1\leq 0.
\end{equation}
This technical condition is used only for the construction of an approximate solution in the next section, and it is not known whether or not it is necessary for existence. It should be pointed out that this extra condition does not restrict the types of horizon topologies that can be produced with our approach, which includes all possibilities \cite[Proposition 3]{KhuriWeinsteinYamada}.
Furthermore, if \eqref{admiss} does not hold then orbifold singularities are allowed and \eqref{compat} should be replaced with the generalized compatibility condition
\begin{equation}\label{compat1}
v_{l-1}^1 v_{l+1}^1 \det
\begin{pmatrix}
v_{l-1}^1 & v_{l-1}^2 \\
v_{l}^1 & v_{l}^2
\end{pmatrix}
\det
\begin{pmatrix}
v_l^1 & v_l^2 \\
v_{l+1}^1 & v_{l+1}^2
\end{pmatrix}
\leq 0.
\end{equation}

With the rod structure and potentials, we may now obtain simple expressions for the charges and angular momenta that characterize stationary bi-axisymmetric solutions. There are two types of such quantities, those which are conserved with respect to homology class and those which are based on Komar integrals. Both will be described.

\subsection{Dipole Charges}

Consider a homology class $[\mathfrak{C}]\in H_2(M^5)$. In the current setting nontrivial classes may be constructed from a single rod $\Gamma_l=[z_l,z_{l-1}]$ and a vector $w\in\mathbb{Z}^2$ in the following way. Let $\Gamma_l$ be either an axis rod bounded by two corners with $w$ linearly independent from the rod structure $v_l$ of this rod, or let $\Gamma_l$ be a ring horizon rod with $w$ the rod structure of the two neighboring axis rods. In the axis case a typical choice for $w$ is $\hat{v}_l=(-v_l^2,v_l^1)^T$, which is perpendicular to $v_l$. In both cases a homology representative $\mathfrak{C}_{w}$, homeomorphic to a 2-sphere, may be constructed by moving the circle associated with $w$ along the rod $\Gamma_l$ from one end point to the other (where it collapses). The dipole charge of this homology class is then given by
\begin{equation}\label{hhhh}
\mathcal{D}_l = \frac{1}{2\pi|w|} \int_{\mathfrak{C}_{w}} \mathcal{F}
=\frac{1}{|w|}\int_{\Gamma_l}\iota_{w^i\eta_{(i)}}\mathcal{F}
=\frac{w^i}{|w|}\left[\psi_i(z_{l-1})-\psi_i(z_l)\right].
\end{equation}

This definition may be generalized to 2-dimensional submanifolds with boundary that are associated with a rod. In particular the same definition and computation apply
if $\Gamma_l$ is an arbitrary axis rod, or a horizon rod with arbitrary $w\in\mathbb{Z}^2$. In this general situation the surface $\mathfrak{C}$ is obtained by moving the circle associated with $w$ in the torus fibers $U(1)^2$ along $\Gamma_l$, and is not necessarily a 2-sphere. Depending on how $w^i\eta_{(i)}$ vanishes at the end points of the rod $\mathfrak{C}$ could be either a disk, cylinder, or sphere.

A dipole charge may also be computed for 2-cycles that are not associated with
a single rod. For example let $\Gamma_{l}=[z_{l},z_{l-1}]$, $l=1,2,3$ be a consecutive sequence of three rods in which the first and third are axis rods and the second is a horizon rod. Consider a semi-circle in the 2-dimensional orbit space connecting the corner point $z_0$ to a point on $\Gamma_3$. The $S^1$ associated with the rod structure $v_3$ may be moved along this curve to produce a 2-sphere. This yields a dipole charge in the same manner as \eqref{hhhh}.

\subsection{Electric Charge}

The total electric charge contained within the spacetime is defined to be
\begin{equation}\label{lala}
\mathcal{Q} = \frac{1}{16\pi} \int_{S_\infty} \left(\star \mathcal{F} + \frac{1}{\sqrt{3}} \mathcal{A} \wedge \mathcal{F} \right),
\end{equation}
where $S_\infty$ represents the limit as $r=\sqrt{\rho^2 +z^2}\rightarrow\infty$ of cross-sectional surfaces $S_r$ at spatial infinity. The quantity \eqref{lala} is
sometimes referred to as the Page charge \cite{Marolf}.
The 3-form integrand is closed as a direct result of the Maxwell equation \eqref{SFEintro}.
Therefore this charge, assuming that the potential $\mathcal{A}$ is globally defined, is conserved in that it is unchanged if $S_\infty$ is replaced by any surface homologous to it. However in general $\mathcal{A}$ will not be globally defined. To avoid this issue we express $\mathcal{Q}$ as an integral over the orbit space of the globally defined potential $\chi$ and apply Stokes' theorem to obtain
\begin{equation}\label{electricchargeformula}
\mathcal{Q} = \frac{\pi}{4} \int_{C_{\infty}} d\chi = \frac{\pi}{4} \left[\chi(\Gamma_1) - \chi(\Gamma_{L+1})\right]
=\frac{\pi}{4}\sum_{l} \left[\chi(z_{l-1}) - \chi(z_{l})\right],
\end{equation}
where $C_{\infty}$ is the semi-circle at infinity in the half plane orbit space (the orientation is taken to be counterclockwise in the $(\rho,z)$ plane).
As computed in the proof of Proposition \ref{proposition98} there is a constant $b_l$ for each axis rod such that
\begin{equation}\label{sqw}
\chi = -\frac{1}{\sqrt{3}|v_l|^2} (\psi \cdot v_l) (\psi \cdot \hat{v}_l) + b_l \quad \text{on} \quad \Gamma_l.
\end{equation}
Furthermore observe that $v^i_l \eta_{(i)}=0$ on $\Gamma_l$, so that $v^i_l d\psi_i=0$ and thus $v^i_l \psi_i=c_l$ is a constant on $\Gamma_l$.
Therefore by working in a gauge such that $c_1=c_{L+1}=0$ we then have that $\chi$ is constant on the two semi-infinite rods so that $\chi(\Gamma_1)$, $\chi(\Gamma_{L+1})$ are well-defined.
Using \eqref{sqw} the expression for total electric charge may be expressed in terms of dipole charges
\begin{equation}\label{kjhyu}
\mathcal{Q} = \frac{\pi}{4}\sum_{l=\text{horizon}} \left[\chi(z_{l-1}) - \chi(z_{l})\right] + \frac{\pi}{4\sqrt{3}} \sum_{l = \text{axis}} \frac{c_l}{|v_l|
} \mathcal{D}_l.
\end{equation}
A consequence of this is that even in the absence of horizons
$\mathcal{Q}$ need not vanish.  It should also be pointed out that
\eqref{lala} is gauge invariant under smooth gauge transformations, but is not necessarily invariant under the so called large gauge transformations \cite{Hanaki:2007mb}.

As explained in the proof of uniqueness in Section \ref{sec8}, it is natural to define an electric horizon charge $\mathcal{Q}_H$ associated with a horizon rod $\Gamma_l$ to be
\begin{equation}
\mathcal{Q}_{H} = \frac{\pi}{4}( b_{l-	1} - b_{l+1}),
\end{equation}
which corresponds to the difference of the constants appearing in \eqref{sqw} that arise from the two surrounding axis rods.  This notion is the direct generalization of horizon charge from stationary axisymmetric solutions of 4D Einstein-Maxwell theory, since it is determined by the change in potential constants across the
horizon rod. Moreover it allows the total charge \eqref{lala} to be expressed as
a combination of horizon charges and dipole charges. This follows from \eqref{kjhyu} by computing the difference in $\chi$. To see this let us consider a specific example of a rod structure with five rods, that is $L=4$, in which $\Gamma_3$ is a horizon rod and the rest are axis rods. Note that the rod point $z_1$ and $z_4$ separate the two semi-infinite rods from the finite rods. We then have
\begin{align}
\begin{split}
\mathcal{Q} =&\frac{\pi}{4}\left[ \chi(z_1) - \chi(z_4)\right]\\
 =& \frac{\pi}{4} \left(-\frac{c_2}{\sqrt{3}|v_2|^2} (\psi(z_1) \cdot \hat{v}_2) + b_2 + \frac{c_4}{\sqrt{3}|v_4|^2} (\psi(z_4) \cdot \hat{v}_4) - b_4\right)  \\
 =& \mathcal{Q}_H - \frac{\pi}{4\sqrt{3}} \left( \frac{c_2}{|v_2|^2} \psi(z_1) \cdot \hat{v}_2 - \frac{c_4}{|v_4|^2 }\psi(z_4) \cdot \hat{v}_4 \right).
\end{split}
\end{align}
As explained in Section \ref{sec8} the quantities $c_l$ and the values of the potentials $\psi_i(z_l)$ at corner points are uniquely determined by the dipole charges. Therefore the total charge agrees with the sum of horizon charges $\mathcal{Q}_H$ up to a combination of dipole fluxes.

A second commonly used definition of electric charge is based on the classical expression from Maxwell's theory
\begin{equation}
Q = \frac{1}{16\pi} \int_{S_\infty} \star \mathcal{F}.
\end{equation}
Note that if $\mathcal{A}\rightarrow 0$ sufficiently fast at infinity then $Q = \mathcal{Q}$. This, however, is not always the case and the difference arises when applying Stokes' theorem  to rewrite $Q$. In particular, let $\Sigma$ denote the $t=0$ slice with boundary $H=\partial\Sigma$ then
\begin{equation}\label{pcharge}
Q = -\frac{1}{16\sqrt{3} \pi}\int_\Sigma \mathcal{F} \wedge \mathcal{F} + \frac{1}{16\pi} \int_H \star \mathcal{F}.
\end{equation}
In the case of solitons, $H=\emptyset$ but the volume integral does not vanish
in general.

Like the conserved charge \eqref{lala} the classical charge \eqref{pcharge} may
also be computed in terms of potentials. To see this let $\varepsilon^{ij}$ be the totally antisymmetric symbol in 2 dimensions with $\varepsilon^{12}=1$ and observe that
\begin{align}\label{mnmn}
\begin{split}
-\frac{1}{4\pi^2}\int_\Sigma \mathcal{F} \wedge \mathcal{F} =&
-\frac{1}{2}\int_{\hat{M}^2} \varepsilon^{ij}\iota_{\eta_{(j)}}
\iota_{\eta_{(i)}}\mathcal{F} \wedge \mathcal{F}\\
=&\int_{\hat{M}^2}\varepsilon^{ij} d\psi_i \wedge d\psi_j\\
=&-\sum_l \int_{\Gamma_l}\varepsilon^{ij} \psi_i d\psi_j
+\int_{C_{\infty}}\varepsilon^{ij} \psi_i d\psi_j,
\end{split}
\end{align}
where $C_{\infty}$ is the semi-circle at infinity in the half-plane orbit space.
Let $\hat{\psi}=(-\psi_2,\psi_1)^T$ and note that from \eqref{klfdkjl} below we have
\begin{equation}\label{fvfv}
\varepsilon^{ij} \psi_i d\psi_j=\hat{\psi}\cdot d\psi
=|v_l|^{-2}(\psi\cdot v_l)d(\psi\cdot\hat{v}_l).
\end{equation}
This shows that the axis rod integrals of \eqref{mnmn} reduce to the difference of values of the potentials at the end points, which in turn is related to the dipole
charge of such rods. Furthermore, the horizon rod integrals of \eqref{mnmn} combine with the horizon integral of \eqref{pcharge} to give $\mathcal{Q}$. Putting this all together yields
\begin{align}
\begin{split}
Q  =& \frac{\pi}{4} \sum_{l = \text{horizon}}\left[ \chi(z_{l-1}) - \chi(z_l)\right] + \frac{\pi}{4\sqrt{3}} \sum_{l=\text{axis}} \frac{c_l}{|v_l|} \mathcal{D}_l +\frac{\pi}{4\sqrt{3}}\int_{C_{\infty}}\varepsilon^{ij} \psi_i d\psi_j \\
=& \mathcal{Q}  +\frac{\pi}{4\sqrt{3}}\int_{C_{\infty}}\varepsilon^{ij} \psi_i d\psi_j .
\end{split}
\end{align}
Under reasonable conditions the asymptotic decay at infinity will imply that
the integral over $C_\infty$ vanishes. Therefore this formula indicates that at least in a gauge in which $\mathcal{A} \to 0$ at spatial infinity, we have $Q = \mathcal{Q}$. Lastly we note that a similar result demonstrating the relation with dipole fluxes was obtained in \cite{Gunasekaran:2016nep} for solitons.

\subsection{Angular Momenta}

The total angular momenta contained within the spacetime is given by the Kormar-type integral
\begin{equation}\label{pppp}
\mathcal{J}_i = \frac{1}{16\pi}\int_{S_\infty} \star d \eta_{(i)} + \mathcal{A}(\eta_{(i)})\left(\star \mathcal{F} + \frac{2}{3\sqrt{3}} \mathcal{A} \wedge \mathcal{F}\right).
\end{equation}
When $\mathcal{F} \equiv 0$ this reduces to the usual definition of Komar angular momenta. The second term has been included in order to render the integrand a closed 3-form. As with the electric charge, however, the presence of the gauge potential $\mathcal{A}$ implies that the integrand need not be globally defined. In order to avoid this we express $\mathcal{J}_i$ as an integral over the orbit space and apply Stokes' theorem to find
\begin{equation}\label{conservedJ}
\mathcal{J}_i =\frac{\pi}{4} \int_{C_{\infty}} d\zeta_i = \frac{\pi}{4} \left[\zeta_i(\Gamma_1) - \zeta_i(\Gamma_{L+1})\right]
=\frac{\pi}{4}\sum_{l} \left[\zeta_i(z_{l-1}) - \zeta_i(z_{l})\right].
\end{equation}
As computed in the proof of Proposition \ref{proposition98} there are constants $a_l$, $\hat{a}_l$ for each axis rod such that
\begin{equation}\label{bgtr}
\zeta = \left(\frac{2}{3\sqrt{3}}\frac{c_l^2 (\psi\cdot\hat{v}_l)}{|v_l|^3}+a_l\right)
\frac{v_l}{|v_l|}
+\left(\frac{1}{3\sqrt{3}}\frac{c_l (\psi\cdot\hat{v}_l)^2}{|v_l|^3}+\hat{a}_l\right)
\frac{\hat{v}_l}{|v_l|} \quad \text{on} \quad \Gamma_l.
\end{equation}
Note that the values $\zeta(\Gamma_1)$, $\zeta(\Gamma_{L+1})$ are well-defined, since working in a gauge such that $c_1=c_{L+1}=0$ yields that $\zeta$ is constant on the two semi-infinite rods. In analogy with electric charge,
\eqref{conservedJ} and \eqref{bgtr} imply that the total angular momentum vector may be written in terms of horizon angular momentum plus an expression determined by dipole charges
\begin{equation}
\mathcal{J} = \sum_{l=\text{horizon}} \mathcal{J}_{l} +\mathfrak{D},
\end{equation}
where $\mathfrak{D}$ depends solely on dipole charges of bubbles and the horizon angular momentum vector associated with a horizon rod $\Gamma_l$ is defined by
\begin{equation}
\mathcal{J}_{l} = \frac{\pi}{4}\left( a_{l-1} \frac{v_{l-1}}{|v_{l-1}|} - a_{l+1} \frac{v_{l+1}}{|v_{l+1}|} +  \hat{a}_{l-1} \frac{\hat{v}_{l-1}}{|v_{l-1}|} - \hat{a}_{l+1} \frac{\hat{v}_{l+1}}{|v_{l+1}|}\right).
\end{equation}
Note that this notion reduces to the typical expression of horizon angular momentum in the vacuum case, which is given by the difference of potential constants on either side of the horizon rod. Moreover, this definition is naturally motivated by its role in the proof of uniqueness in Section \ref{sec8}.

It is also common in the literature to use the standard definition of Komar angular momenta
\begin{equation}
J_i = \frac{1}{16\pi}\int_{S_\infty} \star d \eta_{(i)},
\end{equation}
which is gauge invariant but not conserved between homologous surfaces. By applying Stokes' theorem we obtain
\begin{equation}
J_i = \frac{1}{8\pi}\int_\Sigma \star \text{Ric}(\eta_{(i)}) + \frac{1}{16\pi} \int_H  \star d  \eta_{(i)}.
\end{equation}
Introduce now the closed 2-forms
\begin{equation}
\mathcal{B}_i = \frac{1}{2} \iota_{\eta_{(i)}} \star \mathcal{F} -\frac{1}{\sqrt{3}}\psi_i\mathcal{F} ,
\end{equation}
and observe that the field equations imply
\begin{equation}\label{mnbv}
\frac{1}{8\pi}\int_\Sigma \star \text{Ric}(\eta_{(i)}) = \frac{1}{24\pi} \int_\Sigma \left( d \left(\psi_i\star  \mathcal{F} \right)-\mathcal{B}_i \wedge \mathcal{F} \right).
\end{equation}
The first term on the right-hand side is exact and may be computed on $H$, assuming proper asymptotic decay at infinity. It is then the second term that gives nonzero Komar angular momentum for soliton spacetimes.

Now define the potentials by
\begin{equation}
d\kappa_{ij} = \iota_{\eta_{(i)}} \mathcal{B}_j,
\end{equation}
and note that $v^i\kappa_{ij}$ are constants along an axis on which $v^i \eta_{(i)}$ vanishes.  In terms of the harmonic map potentials
\begin{equation}\label{dj}
d\kappa_{ij} = -\frac{\varepsilon_{ij}}{2} \left [d \chi + \frac{1}{\sqrt{3}}\left(\psi_1 d\psi_2 - \psi_2 d\psi_1\right)\right] - \frac{1}{\sqrt{3}} \psi_j d\psi_i.
\end{equation}
We now compute the expression from \eqref{mnbv} in terms of these new quantities
\begin{equation}\label{JVolterm}
-\frac{1}{4\pi^2} \int_\Sigma \mathcal{B}_i \wedge \mathcal{F} = \int_{\hat{M}^2} \varepsilon^{jm} d\kappa_{ji} \wedge d \psi_m =  \int_{\hat{M}^2} d \left[ \varepsilon^{jm} \kappa_{ji} \wedge d \psi_m\right] =  \int_{\hat{M}^2} d \left[ \varepsilon^{mj} \psi_m  \wedge d\kappa_{ji} \right].
\end{equation}
Observe that the integrand has been written as a total derivative in two alternate forms in order to obtain desirable expressions for the cases of spacetimes with and without horizons.  Let us assume first that the solution is a soliton, that is it does not contain any black holes, we then have
\begin{equation}
J_i  = - \frac{\pi}{6}  \sum_{l}\int_{\Gamma_l} \varepsilon^{jm} \kappa_{ji} d\psi_m +\frac{\pi}{6}\int_{C_\infty}\varepsilon^{jm} \kappa_{ji} d\psi_m.
\end{equation}
Since $v_l \cdot \psi = c_l$ is constant on $\Gamma_l$, a similar calculation to that of \eqref{fvfv} implies
\begin{equation}
\varepsilon^{jm} \kappa_{ji} d\psi_m
= |v_l|^{-2} (v_l \cdot \kappa_i) d(\hat{v}_l \cdot \psi).
\end{equation}
Furthermore it also holds that $v_l \cdot \kappa_{i} = \mathbf{d}_{li}$ is  constant on $\Gamma_l$, and therefore
\begin{equation}
J_i   = \frac{\pi}{6} \sum_{l}|v_l|^{-1} \mathbf{d}_{li} \mathcal{D}_l
+\frac{\pi}{6}\int_{C_\infty}\varepsilon^{jm} \kappa_{ji} d\psi_m.
\end{equation}
This shows that $J_i$, in contrast to $\mathcal{J}_i$, can be nonzero for soliton spacetimes, with a value given as a weighted sum over dipole charges if the asymptotic decay at infinity guarantees that the integral over $C_{\infty}$ vanishes.

Consider now the case in which the spacetime contains black hole horizons. In this situation the last integrand in \eqref{JVolterm} yields
\begin{equation}
-\frac{1}{24\pi} \int_{\Sigma} \mathcal{B}_i \wedge \mathcal{F} = -\frac{\pi}{6} \sum_{l}  \int_{\Gamma_l} \varepsilon^{mj} \psi_m  d \kappa_{ji} + \frac{\pi}{6}\int_{C_\infty}\varepsilon^{mj} \psi_m  d \kappa_{ji}.
\end{equation}
On an axis rod we have
\begin{equation}
\varepsilon^{mj} \psi_m d \kappa_{ji} = |v_l|^{-2} c_l d (\hat{v}_l \cdot \kappa_i),
\end{equation}
whereas on a horizon rod \eqref{dj} gives
\begin{equation}
\varepsilon^{mj} \psi_m  d \kappa_{ji} = \frac{\psi_i}{2} \left( d\chi - \frac{1}{\sqrt{3}} \left(\psi_1 d\psi_2 - \psi_2 d\psi_1\right)\right).
\end{equation}
It follows that the relation between the two notions of angular momentum is given by
\begin{align}
\begin{split}
J_i =& \frac{\pi}{6} \sum_{l=\text{axis}} \frac{c_l}{|v_l|^2} \hat{v}_l^j \left[\kappa_{ji}(z_{l}) - \kappa_{ji}(z_{l-1})\right]+ \frac{\pi}{6}\int_{C_\infty}\varepsilon^{mj} \psi_m  d \kappa_{ji}\\
&
-\frac{\pi}{6} \sum_{l = \text{horizon}} \int_{\Gamma_l}  \varepsilon^{mj} \psi_m d \kappa_{ji} + \frac{1}{16\pi} \int_H \left(\star d \eta_{(i)}  - \frac{2}{3} \psi_i \star \mathcal{F}\right). 
\end{split}
\end{align}
If we associate a dipole-like charge to the flux of $\mathcal{B}_i$ out of 2-surface $\mathfrak{C}$ by setting
\begin{equation}
\mathcal{K}_i = \frac{1}{2\pi |w|}\int_{\mathfrak{C}} \mathcal{B}_i = \frac{w^j}{|w|}\left[\kappa_{ji}(z_{l}) - \kappa_{ji}(z_{l-1})\right],
\end{equation}
and use $\iota_{\eta_{(2)}}\iota_{\eta_{(1)}} \star d \eta_{(i)} = \Theta_i$ and \eqref{eq2.11} then the final angular momentum expression takes the form
\begin{equation}\label{fgsd}
J_i = \frac{\pi}{4} \sum_{l =\text{horizon}} \left[\zeta_i(z_{l-1}) - \zeta_i(z_{l})\right] + \frac{\pi}{6} \sum_{l = \text{axis}} \frac{c_l}{|v_l|} \mathcal{K}_{li}
+ \frac{\pi}{6}\int_{C_\infty}\varepsilon^{mj} \psi_m  d \kappa_{ji}.
\end{equation}

We now find the relation between the two definitions of total angular momenta.
Observe that on an axis rod $\Gamma_l$
\begin{equation}
d\kappa_{ij} = -\frac{1}{\sqrt{3}} \psi_j d\psi_i,
\end{equation}
and therefore
\begin{equation}
v^i_l d \kappa_{ij}  = -\frac{1}{\sqrt{3}} \psi_j d (\psi\cdot v_l) =0,\quad\quad
\hat{v}^i_l d\kappa_{ij} = -\frac{1}{\sqrt{3}} \psi_j d (\psi\cdot \hat{v}_l).
\end{equation}
It follows from \eqref{bgtr} that on the axis
\begin{equation}
d\zeta_k = -\frac{2c_l}{3|v_l|^4} \left(  d (v^j_l \hat{v}^i_l \kappa_{ij})v^k_l  + d(\hat{v}^j_l \hat{v}^i_l \kappa_{ij})\hat{v}^k_l\right),
\end{equation}
and hence
\begin{equation}
d\zeta_j = -\frac{2c_l}{3|v_l|^2} \hat{v}^i_l d\kappa_{ij}.
\end{equation}
Using \eqref{conservedJ} and \eqref{fgsd} then yields
\begin{align}
\begin{split}
\mathcal{J}_i =& \frac{\pi}{4} \sum_{l =\text{horizon}} \left[\zeta_i(z_{l-1}) - \zeta_i(z_{l})\right] + \frac{\pi}{6} \sum_{l = \text{axis}} \frac{c_l}{|v_l|} \mathcal{K}_{li}\\
=& J_i -\frac{\pi}{6}\int_{C_\infty}\varepsilon^{mj} \psi_m  d \kappa_{ji}.
\end{split}
\end{align} Under appropriate asymptotic decay conditions on the potentials $\psi_i$, the second term will vanish.

\section{The Approximate Solution}
\label{sec6}\setcounter{equation}{0}
\setcounter{section}{6}

In this section we will begin the process of solving the harmonic map equations \eqref{HMap} with prescribed
rod structure and potentials on the axes. Our approach is motivated by the work of Weinstein \cite{WeinsteinHadamard}, as well as the methods presented in \cite{KhuriWeinsteinYamada,KhuriWeinsteinYamada1}.
The first step is to construct
a type of approximate solution referred to as the model map $\Psi_0:\mathbb{R}^3\setminus\Gamma\rightarrow \tilde{N}$, where $\tilde{N}$ is the set of $7\times 7$ positive definite unimodular matrices which may be used to represent the coset space $N=G_{2(2)}/SO(4)\cong\mathbb{R}^8$. The coset representatives in $\tilde{N}$ are parameterized \cite[\S 3]{Alaee:2018unn} by the coordinates $F=(f_{ij})$, $\zeta=(\zeta_1,\zeta_2)^T$, $\chi$, and $\psi=(\psi_1,\psi_2)^T$, and according to \eqref{yyyy} the canonical complete nonpositively curved metric on $N$ takes the form
\begin{equation}\label{themetric}
4G=\frac{1}{2}\left[\text{Tr}\left(F^{-1}dF\right)\right]^2
+\frac{1}{2}\text{Tr}\left[\left(F^{-1}dF\right)^2\right]
+f^{-1}\Theta^T F^{-1}\Theta
+f^{-1}\Upsilon^2
+d\psi^T F^{-1}d\psi,
\end{equation}
where $f=\det F$, and $\Theta$, $\Upsilon$ are given by \eqref{eq2.11}, \eqref{uuuu}. Recall that the tension of a map between two Riemannian manifolds $\varphi:M\rightarrow N$ is a section of the pullback bundle $\varphi^{*}TN$ and is given by
\begin{equation}
\tau(\varphi)=\tilde{\nabla}^a \partial_a \varphi,
\end{equation}
where $\tilde{\nabla}$ is the induced connection on $T^{*}M\otimes \varphi^{*}TN$. The tension field measures how far away a map is from being harmonic, in that $\varphi$ is harmonic if and only if $|\tau(\varphi)|=0$. In the current setting the components of the tension are found to be
\begin{equation}
\tau=\tau^{f_{ij}}\partial_{f_{ij}}+\tau^{\Theta_i}{\Theta}_i
+\tau^{{\Upsilon}}{\Upsilon}+\tau^{\psi_{i}}\partial_{\psi_{i}}
\end{equation}
where
\begin{align}
\begin{split}
F^{-1}\tau^{F}=&\mathrm{div} \left(F^{-1} dF\right)+f^{-1}F^{-1}\Theta \cdot \Theta^T+\frac{1}{3} f^{-1}\Upsilon\cdot\Upsilon I_{2}\\
&+\frac{1}{3}F^{-1}d\psi\cdot d\psi^T -\frac{1}{3}d\psi^T \cdot\left(F^{-1}d\psi\right) I_{2},\\
F^{-1}\tau^{\Theta}=&f\mathrm{div} \left(f ^{-1}F^{-1} \Theta\right),\\
\tau^{\Upsilon}=&f\mathrm{div}\left(f^{-1}\Upsilon \right)+d\psi^T \cdot\left(F^{-1}\Theta\right),\\
F^{-1}\tau^{\psi}=&\mathrm{div} (F^{-1} d\psi)-f^{-1}\Upsilon\cdot\left(F^{-1}\Theta\right)
+\frac{2f^{-1}}{\sqrt{3}}\left( \delta_{2} \Upsilon d\psi^1 - \delta_{1} \Upsilon d\psi^2 \right),
\end{split}
\end{align}
in which $I_{2}$ is identity $2\times 2$ matrix and all inner products are with respect to the flat metric. It follows that the norm squared of the tension is
\begin{align}\label{tensionnorm}
\begin{split}
4|\tau|^2=&\frac{1}{2}\left[\text{Tr}\left(\mathrm{div} \left(F^{-1} dF\right)+\frac{F^{-1}\Theta\cdot\Theta^T}{f}+\frac{F^{-1}d\psi\cdot d\psi^T}{3} -\frac{d\psi^T \cdot\left(F^{-1}d\psi\right) I_{2}}{3}
+\frac{|\Upsilon|^2 I_{2}}{3f}\right)\right]^2\\
+&\frac{1}{2}\text{Tr}\left[\left(\mathrm{div} \left(F^{-1} dF\right)+\frac{F^{-1}\Theta\cdot\Theta^T}{f}+\frac{F^{-1}d\psi\cdot d\psi^T}{3} -\frac{d\psi^T \cdot\left(F^{-1}d\psi\right) I_{2}}{3}
+\frac{|\Upsilon|^2 I_{2}}{3f}\right)^2\right]\\
+&f\left[\mathrm{div} \left(f^{-1}F^{-1} \Theta\right)\right]^T F \left[\mathrm{div} \left(f^{-1}F^{-1} \Theta\right)\right]+f\left[\mathrm{div}\left(f^{-1}\Upsilon \right)+f^{-1}d\psi^T \cdot\left( F^{-1}\Theta\right)\right]^2\\
+&\left[\mathrm{div}(F^{-1} d\psi)-f^{-1}\Upsilon \cdot \left(F^{-1}\Theta\right)+\frac{2f^{-1}}{\sqrt{3}}\left( \delta_{2} \Upsilon \cdot d\psi_1 - \delta_{1} \Upsilon\cdot d\psi_2 \right)\right]^T\\
\times & F\left[\mathrm{div} (F^{-1} d\psi)-f^{-1}\Upsilon\cdot \left(F^{-1}\Theta\right)+\frac{2f^{-1}}{\sqrt{3}}\left( \delta_{2} \Upsilon\cdot d\psi_1 - \delta_{1} \Upsilon\cdot d\psi_2 \right)\right].
\end{split}
\end{align}

In order for a model map $\Psi_0$ to be considered an appropriate approximate solution on which to build the existence theory it must keep the tension bounded and properly decaying at infinity, as well as share the same rod structure and potential constants along the axes as those that are prescribed for the solution.

\begin{proposition}\label{proposition98}
Let $\Gamma_l$ be a set of axis rods having corresponding rod structures $(v_l^1,v_l^2)$ satisfying the compatibility condition \eqref{compat1}, and let $\mathbf{a}_l$, $\mathbf{b}_l$,
$\mathbf{c}_l$ be a set of associated constants in which only $\mathbf{c}_l$ may change between rods that share an end point. Then there exists a
model map $\Psi_0:\mathbb{R}^3\setminus\Gamma\rightarrow \tilde{N}$ that possesses uniformly bounded tension, decays at infinity by $|\tau|=O(r^{-3})$, and satisfies
$v_{l}^i \psi_i=\mathbf{c}_l$ on $\Gamma_l$, with $\zeta$ and $\chi$ agreeing with $\mathbf{a}_l$ and $\mathbf{b}_l$ on $\Gamma_l$ up to a function depending only on $\psi$ and the rod structure. Furthermore the functions defining the model map $(F, \zeta,\chi,\psi)$ are all smooth everywhere including along the axis and at corners.
\end{proposition}

\begin{remark}
Given a rod structure and corresponding model map provided by this proposition,
the constants $\mathbf{a}_l$, $\mathbf{b}_l$, $\mathbf{c}_l$ may be used to prescribe the angular momenta $\mathcal{J}_i$ and electric charge $\mathcal{Q}$ of each horizon rod, as well as the dipole charge $\mathcal{D}$ for each axis rod whose end points are corners. In fact, the angular momenta and electric charge of a horizon component are simply (up to a constant multiple) the difference of the constants $\mathbf{a}_l$ and $\mathbf{b}_l$ on each side of the relevant horizon rod.
\end{remark}

\begin{proof}
Consider three domains whose disjoint union is $\mathbb{R}^3=D_1\cup D_2\cup D_3$. Let $D_1=\mathbb{R}^3\setminus B_{r_0}$ be the complement of a large ball which intersects the two semi-infinite rods, and let $D_2$ be a small tubular neighborhood of the axis rods inside $B_{r_0}$. The domain $D_3$ is then the complement of $D_2$ within $B_{r_0}$. This decomposition is depicted in Figure \ref{figure2}.
Consider first the case in which no connected component of the axis $\Gamma$ has more than one corner. By setting the potentials $(\zeta,\chi,\psi)$ to be the appropriate prescribed constants on connected components of $D_2$, the tension norm $|\tau|$ reduces to the same expression as that in the vacuum case treated in \cite[Theorem 6]{KhuriWeinsteinYamada}, and thus the definition of $F$ in $D_2$ is taken to be the same as given there. The tension is then bounded in this domain. Suppose further that the model map $\Psi_0$ is given in $D_1$, then in $D_3$ we may set it to be any function which interpolates smoothly between the definitions in $D_1$ and $D_2$.

Let us now construct the model map in the exterior region $D_1$. On this domain
define
\begin{equation}
F=\begin{pmatrix}
r\sin^2(\theta/2) & 0 \\
0 & r\cos^2(\theta/2)
\end{pmatrix},
\qquad \zeta=\zeta(\theta),\quad \chi=\chi(\theta), \quad \psi=\psi(\theta),
\end{equation}
where $(r,\theta)$ are polar coordinates i.e. $\rho=r\sin\theta$, $z=r\cos\theta$.
The components of $F$ are harmonic functions and therefore $\operatorname{div} \left(F^{-1} dF\right)=0$. In addition, since $f_{11}$ behaves like $2\log\rho$ near the positive $z$-axis and is bounded near the negative $z$-axis while $f_{22}$ has the opposite behavior, the rod structure arising from this prescription is $(1,0)$ for the northern semi-infinite rod and $(0,1)$ for the southern semi-infinite rod. This is the rod structure associated with an asymptotically flat spacetime. Next, the potential functions are chosen to be the
appropriate prescribed constants near the axes, that is for
$\theta\in[0,\varepsilon]\cup[\pi-\varepsilon,\pi]$ with $0<\varepsilon$ small. It follows that near the axes in $D_1$ the model map is harmonic so that $|\tau|=0$.
We may now choose $(\zeta,\chi,\psi)$ to be arbitrary smooth functions of $\theta$
that interpolate between the two sets of constants for $\theta\in[\varepsilon,\pi-\varepsilon]$.
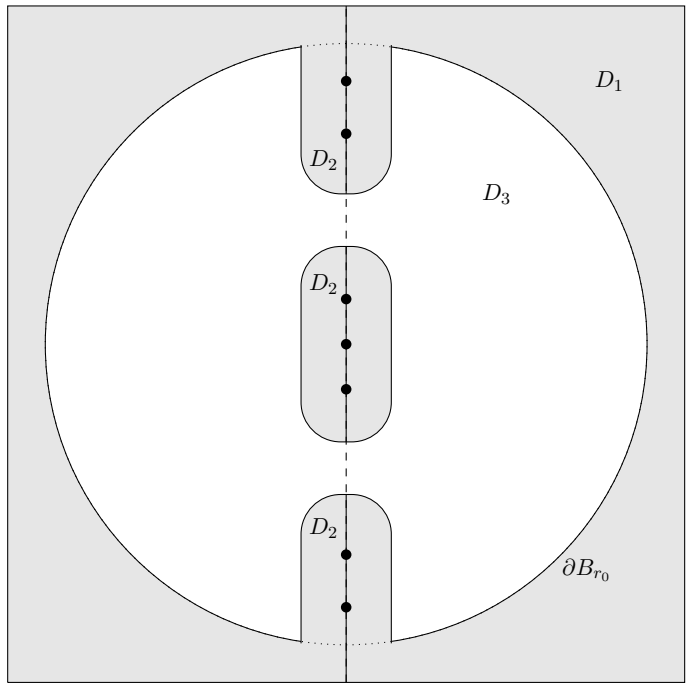
\begin{figure}[h]	
	\begin{tikzpicture}[scale=1, every node/.style={scale=0.8}]
	 \draw[fill=gray!20!white](-4.5,-4.5)--(4.5,-4.5)--(4.5,4.5)--(-4.5,4.5)--(-4.5,-4.5);
	\draw[black,fill=white] (0,-4cm) arc (-90:90:40mm);
	\draw[black,fill=white] (0,4cm) arc (90:270:40mm);
	\draw[rounded corners=15pt,fill=gray!20!white] (-0.6,3.98)--(-0.6,2)--(0.6,2)--(0.6,3.98);
	\draw[rounded corners=15pt,fill=gray!20!white] (-0.6,-3.98)--(-0.6,-2)--(0.6,-2)--(0.6,-3.98);
	\draw[rounded corners=15pt,fill=gray!20!white] (-0.6,0)--(-0.6,1.3)--(0.6,1.3)--(0.6,0)--(0.6,-1.3)--(-0.6,-1.3)--(-0.6,0);
	\fill[gray!20!white](-0.5,-3.7)--(-0.5,-4.2)--(0.5,-4.2)--(0.5,-3.7);
	\fill[gray!20!white](-0.5,3.7)--(-0.5,4.2)--(0.5,4.2)--(0.5,3.7);
	\draw[dashed](0,-4.5)--(0,4.5);
	\draw[](0,-2)--(0,-4.5);
	\draw[](0,2)--(0,4.5);
	\draw[](0,-1.3)--(0,1.3);
	\fill[black] (0,2.8) circle [radius=.07];
	\fill[black] (0,3.5) circle [radius=.07];
	\fill[black] (0,0.6) circle [radius=.07];
	\fill[black] (0,0) circle [radius=.07];
	\fill[black] (0,-0.6) circle [radius=.07];
	\fill[black] (0,-2.8) circle [radius=.07];
	\fill[black] (0,-3.5) circle [radius=.07];
	\node at (-.3,2.45){$D_2$};
	\node at (-.3,-2.45){$D_2$};
	\node at (-.3,.8){$D_2$};
	\node at (3.5,3.5){$D_1$};
	\node at (2,2){$D_3$};
	\node at (3.2,-3){$\partial B_{r_0}$};
	\draw[dotted,black] (0,-4cm) arc (-90:90:40mm);
	\draw[dotted,black] (0,4cm) arc (90:270:40mm);
	\end{tikzpicture}
	\caption{Domain Decomposition}\label{figure2}
\end{figure}
It will now be shown that the tension $|\tau|$ decays like $O(r^{-3})$. According to the description above the tension vanishes near the axes, and so this condition need only be checked on the interpolation region. There, using the explicit description of $F$, the asymptotics for each term may be computed. For instance, consider the second term on the right-hand side of \eqref{tensionnorm}. The portion $F^{-1}$ decays like $O(r^{-1})$, $f^{-1}$ decays like $O(r^{-2})$, and the inner product contributes an extra $O(r^{-2})$ since the derivatives within $\Theta$ are only in the $\theta$ direction. It follows that
\begin{equation}
f^{-1}F^{-1}\Theta\cdot\Theta^{T}=O(r^{-5}).
\end{equation}
Similar considerations may be applied to each term yielding $|\tau|=O(r^{-3})$.
For the ALE and AKK asymptotics the model map construction is the same except that $F$ is modified appropriately in the region $D_1$, see \S 4.1 and \S 4.2 of \cite{KhuriWeinsteinYamada1} respectively. Analogous arguments may then be made to estimate the asymptotics of each term appearing in \eqref{tensionnorm} to arrive at the same conclusion.

It remains to define the model map in the region $D_2$ when components of $\Gamma$ have more than one corner. For each component, the construction may be accomplished inductively on the number of rods. Thus we will give details only for a sequence of three rods separated by two corners. Consider a consecutive sequence of axis rods: the north $\Gamma_1$, middle $\Gamma_2$, and south $\Gamma_3$ having rod structures $v_l=(v^1_l,v^2_l)$, $l=1,2,3$, and separated by corners $p_{1}$, $p_2$. It may be assumed without loss of generality for the purposes here that the rod structures are of unit norm $|v_l|=1$. Let $D$ denote the region of $D_2$ which contains these rods. The construction of $F$ in this domain follows that of \cite[Theorem 6]{KhuriWeinsteinYamada}. Namely by choosing appropriate harmonic functions $u$ and $v$ the matrix
\begin{equation}
\bar{F}=\begin{pmatrix} e^u & 0 \\ 0 & e^v \end{pmatrix},
\end{equation}
gives rise to rod structure $(1,0)$ on $\Gamma_1 \cup\Gamma_3$, and $(0,1)$ on $\Gamma_2$. The model map matrix is then set to $F=\mathrm{k} \bar{F}\mathrm{k}^T$ where $\mathrm{k}=\mathrm{k}(z)$ is a $2\times 2$ nonsingular matrix function that is constant except on the interior of $\Gamma_2$, where it transitions smoothly between rod structures. This definition realizes the desired rod structures and has the property that $\operatorname{div}\left(F^{-1}dF\right)$ is uniformly bounded in $D$. The above construction is motivated by the fact that
\begin{equation}
F\mapsto \mathrm{k}F \mathrm{k}^T,\quad \psi\mapsto \mathrm{k}\psi,\quad \chi\mapsto (\det \mathrm{k})\chi,\quad
\zeta\mapsto (\det \mathrm{k}) \mathrm{k}\zeta,
\end{equation}
\begin{equation}
\Rightarrow \quad\quad \Upsilon\mapsto (\det \mathrm{k})\Upsilon,\quad\quad \Theta\mapsto (\det \mathrm{k})\mathrm{k}\Theta,
\end{equation}
is an isometry of the target space.

We will now define the potentials of the model map in $D$; they will all be functions of $z$ alone. The magnetic potential $\psi$ is defined to be a smooth function which satisfies
\begin{equation}\label{psiconstant}
v_l^i\psi_i=c_l\quad\text{ on }\quad\Gamma_l,
\end{equation}
for arbitrary constants $c_l=\mathbf{c}_l$. Note that this fixes one component of $\psi$ on each rod while the other component is allowed to transition. Furthermore, since neighboring rod structures are linearly independent the prescription \eqref{psiconstant} fully determines $\psi$ at the corners, and $\psi$ may be
taken to be this constant quantity in a neighborhood of the corners.

Consider next the electric potential. The goal is to choose $\chi$ so that
\begin{equation}\label{tt}
\Upsilon=d\chi+\frac{1}{\sqrt{3}}\hat{\psi}\cdot d\psi=0\quad\text{ on }\quad D,
\end{equation}
where the `hat' operation takes a vector to one which is orthogonal to the original
and having the same norm $\hat{\psi}=(-\psi_2,\psi_1)^T$. By using \eqref{psiconstant} and the fact that $v_l$ and $\hat{v}_l$ are constant on $\Gamma_l$ we have
\begin{equation}\label{klfdkjl}
\hat{\psi}\cdot d\psi = |v_l|^{-4}\left[-(\psi\cdot\hat{v}_l)v_l+(\psi\cdot v_l)\hat{v}_l\right]
\cdot d\left[(\psi\cdot v_l) v_l+(\psi\cdot\hat{v}_l)\hat{v}_l\right]\\
= |v_l|^{-2}d\left[(\psi\cdot v_l)(\psi\cdot\hat{v}_l)\right].
\end{equation}
It follows that $\Upsilon=0$ on each $\Gamma_l$ by setting
\begin{equation}\label{chiformula}
\chi=-\frac{1}{\sqrt{3}}(\psi\cdot v_l)(\psi\cdot \hat{v}_l)+b_l\quad\text{ on }\quad\Gamma_l,
\end{equation}
where $b_l$ are arbitrary constants. Furthermore
since $\psi$ is constant in a neighborhood of the corners, the function $\chi$ may be made continuous by appropriately choosing the constants $b_2-b_1$ and $b_3-b_2$. Then among the three constants $b_l$ there is one left that may be chosen arbitrarily, and so the smooth function $\chi$ is defined up
to a single constant $\mathbf{b}_l$ on $D$.

Lastly, the twist potentials are chosen to achieve
\begin{equation}\label{zetaconstant}
\Theta=d\zeta+\psi[d\chi+\frac{1}{3\sqrt{3}}\hat{\psi}\cdot d\psi]=0\quad\text{ on }\quad D.
\end{equation}
Using \eqref{psiconstant}, \eqref{tt}, and \eqref{klfdkjl} shows that on $\Gamma_l$
\begin{align}\label{nbnbnbm}
\begin{split}
\Theta=&d\zeta-\frac{2}{3\sqrt{3}}\psi d\left[(\psi\cdot v_l)(\psi\cdot\hat{v}_l)\right]\\
=&d\left[(\zeta\cdot v_l)v_l+(\zeta\cdot\hat{v}_l)\hat{v}_l\right]
-\frac{2}{3\sqrt{3}}\left[(\psi\cdot v_l)v_l+(\psi\cdot\hat{v}_l)\hat{v}_l\right]
d\left[(\psi\cdot v_l)(\psi\cdot\hat{v}_l)\right]\\
=&v_l d\left[(\zeta\cdot v_l)-\frac{2}{3\sqrt{3}}(\psi\cdot v_l)^2(\psi\cdot
\hat{v}_l)\right]
+\hat{v}_l d\left[(\zeta\cdot \hat{v}_l)-\frac{1}{3\sqrt{3}}(\psi\cdot v_l)(\psi\cdot\hat{v}_l)^2\right].
\end{split}
\end{align}
Therefore to achieve \eqref{zetaconstant} define $\zeta$ on $\Gamma_l$ by
\begin{equation}\label{zetaformula}
\zeta\cdot v_l=\frac{2}{3\sqrt{3}}(\psi\cdot v_l)^2(\psi\cdot
\hat{v}_l)+a_l,\quad\quad
\zeta\cdot \hat{v}_l=\frac{1}{3\sqrt{3}}(\psi\cdot v_l)(\psi\cdot\hat{v}_l)^2
+\hat{a}_l,
\end{equation}
for arbitrary constants $a_l$, $\hat{a}_l$. As in the definition of $\chi$, the function $\zeta$ may be made continuous by properly choosing the constants $a_2-a_1$, $a_3-a_2$, $\hat{a}_2-\hat{a}_1$, and $\hat{a}_3-\hat{a}_2$. There is then one degree of freedom left for each component of $\zeta$, and thus $\zeta$ is
defined up to a constant vector $\mathbf{a}_l$ on $D$.

We may now complete the proof. Consider each term in the tension expression \eqref{tensionnorm}. In light of \eqref{tt} and \eqref{zetaconstant} all terms involving $\Upsilon$ and $\Theta$ vanish in $D$. Moreover as mentioned above $\operatorname{div}\left(F^{-1}dF\right)$ is uniformly bounded. Finally \eqref{psiconstant} implies that $F^{-1}d\psi$ and $\operatorname{div}\left(F^{-1}d\psi\right)$ remain bounded as well. Hence $|\tau|$ is properly controlled in $D$. Lastly it is clear from the construction that the degrees of freedom may be chosen properly so that
$v_{l}^i \psi_i=\mathbf{c}_l$ on $\Gamma_l$, with $\zeta$ and $\chi$ agreeing with $\mathbf{a}_l$ and $\mathbf{b}_l$ on $\Gamma_l$ up to a function depending only on $\psi$ and the rod structure; here $\mathbf{a}_l=\mathbf{a}$ is the same constant for $l=1,2,3$ and similarly for $\mathbf{b}_l=\mathbf{b}$.
\end{proof}

\section{Harmonic Map Existence and Uniqueness}
\label{sec7}\setcounter{equation}{0}
\setcounter{section}{7}

With the model map $\Psi_0$ in hand, the proof of the existence and uniqueness of a harmonic map $\Psi:\mathbb{R}^3 \setminus\Gamma\rightarrow \tilde{N}$ which is asymptotic to the model map may now be carried out by following the arguments in
the vacuum case \cite{KhuriWeinsteinYamada} with slight modification. This is possible due to the fact that the target symmetric space here, $N=G_{2(2)}/SO(4)$, is
nonpositively curved and of rank 2 just as the target space in the vacuum case $SL(3,\mathbb{R})/SO(3)$. For the sake of completeness we will sketch the arguments. Recall that two maps are said to be \textit{asymptotic} if the $N$-distance between the two $d(\Psi,\Psi_0)$ remains bounded near the axes, and $d(\Psi,\Psi_0)\rightarrow 0$ as $r\rightarrow \infty$.
As is shown in \cite[Theorem 11]{KhuriWeinsteinYamada}, if $\Psi$ and $\Psi_0$ are asymptotic then they give rise to the same rod structure and the values of the two sets of potentials on the axes agree. Thus, the spacetime resulting from $\Psi$ will have the prescribed rod structure and hence topology, as well as the prescribed charges.

Consider now the question of uniqueness. Let $\Psi_1$ and $\Psi_2$ be two harmonic maps that are asymptotic with the same model map $\Psi_0$.
Since the target space is nonpositively curved it follows \cite[Lemma 2]{WeinsteinHadamard} that
\begin{equation}\label{1999}
\Delta \sqrt{1 + d(\Psi_1,\Psi_2)^2}  \geq
-|\tau(\Psi_1)| - |\tau(\Psi_2)| =0.
\end{equation}
As the two maps are asymptotic to each other there is a uniform bound for the distance $d(\Psi_1,\Psi_2)\leq C$, and we may then interpret the function $\sqrt{1 + d(\Psi_1,\Psi_2)^2}$ as weakly subharmonic on $\mathbb{R}^3$. Due to the fact that $\Gamma$ is of codimension 2, the maximum principle applies \cite[Lemma 8]{Weinstein1} to show that $\sqrt{1 + d(\Psi_1,\Psi_2)^2}\leq 1$, since  $d(\Psi_1,\Psi_2)\rightarrow 0$ at infinity. Hence $\Psi_1=\Psi_2$.

The proof of existence proceeds as follows. Let $\mathfrak{D}_j$ be an increasing sequence of domains that exhaust $\mathbb{R}^3\setminus\Gamma$ as $j\rightarrow\infty$, and let $\Psi_j$ be the unique harmonic map on $\mathfrak{D}_j$ which agrees with the model map on the boundary, that is, having the Dirchlet boundary conditions $\Psi_j=\Psi_0$ on $\partial\mathfrak{D}_j$. Since
$|\tau(\Psi_0)|$ is uniformly bounded and decays sufficiently fast at infinity, there exists \cite[pg. 838]{WeinsteinHadamard} a positive smooth function $w$ on $\mathbb{R}^3$ satisfying $\Delta w\leq -|\tau(\Psi_0)|$ such that $w\rightarrow 0$ as $r\rightarrow\infty$ \cite[Lemma 10]{KhuriWeinsteinYamada}. Then with the aid of \eqref{1999} we find
\begin{equation}\label{f82y-0---}
\Delta\left( \sqrt{1 + d(\Psi_j,\Psi_0)^2} -w\right) \geq 0,\quad\quad\quad
\sqrt{1 + d(\Psi_j,\Psi_0)^2} -w\leq 1\text{ }\text{ on }\text{ }\partial\mathfrak{D}_j.
\end{equation}
The maximum principle may be used again to produce
a uniform $C^0$ estimate for $d(\Psi_j,\Psi_0)$. From this, local pointwise
energy estimates may be established following \cite[Section 6]{KhuriWeinsteinYamada}. Note that although \cite[Section 6]{KhuriWeinsteinYamada}
is written explicitly for the rank 2 target space $SL(3,\mathbb{R})/SO(3)$, analogous arguments may be given by $G_{2(2)}/SO(4)$ since it is also of rank 2.
Standard elliptic bootstrapping can now be implemented to control all higher order derivatives of $\Psi_j$ on compact subsets. Therefore this sequence of maps subconverges to a harmonic map $\Psi$ having a distance to the model map which is uniformly bounded and vanishes at infinity, since \eqref{f82y-0---} implies
\begin{equation}\label{finainfoiqhogihq}
d(\Psi,\Psi_0)\leq \sqrt{w(2+w)}.   
\end{equation}
In particular the solution is asymptotic to $\Psi_0$. We have proved the following result.

\begin{theorem}\label{hmexist}
Given a model map $\Psi_0$ having uniformly bounded tension field that decays at infinity like $O(r^{-3})$, there exists a unique harmonic map $\Psi:\mathbb{R}^3\setminus\Gamma\rightarrow\tilde{N}$ which is asymptotic to $\Psi_0$.
\end{theorem}

Let us now complete the proof of the main theorem. Suppose that a set of rod structure data and corresponding potential constants are given, satisfying the assumptions of Theorem \ref{primaryresult}. By Proposition \ref{proposition98} there exists
a model map $\Psi_0$ which encodes this data and falls-off appropriately at infinity. Theorem \ref{hmexist} may now be applied to find a unique harmonic map $\Psi$ which is asymptotic to $\Psi_0$. From this harmonic map, a solution of the supergravity equations may be constructed according to the description in Section \ref{yyyyy}. Arguments similar to those used in the vacuum case \cite[Theorem 11]{KhuriWeinsteinYamada}, which are detailed below, may now be employed to show that this solution realizes the prescribed rod data and charges.
Additionally, conical singularities may be ruled out on the semi-infinite rods as in \cite[\S 6]{KhuriWeinsteinYamada1}.

\begin{theorem}\label{aoifjikhjgsa}
Let $(F,\zeta,\chi,\psi)$ and $(F_0,\zeta_0,\chi_0,\psi_0)$ denote the functions defining the harmonic map $\Psi$ and model map $\Psi_0$ of Theorem \ref{hmexist}. Then on each axis rod $\Gamma_l$ with rod structure $v_l$, we have $\mathrm{ker} \text{ }F=\mathrm{ker}\text{ } F_0$ and 
\begin{equation}\label{2309i11akhkkae}
\zeta\cdot v_l-\frac{2}{3\sqrt{3}}(\psi\cdot v_l)^2(\psi\cdot
\hat{v}_l)=\zeta_0\cdot v_l-\frac{2}{3\sqrt{3}}(\psi_0\cdot v_l)^2(\psi_0\cdot
\hat{v}_l),
\end{equation}
\begin{equation}\label{2309i11akhkkae1}
\zeta\cdot \hat{v}_l-\frac{1}{3\sqrt{3}}(\psi\cdot v_l)(\psi\cdot\hat{v}_l)^2
=\zeta_0\cdot \hat{v}_l-\frac{1}{3\sqrt{3}}(\psi_0\cdot v_l)(\psi_0\cdot\hat{v}_l)^2,
\end{equation}
\begin{equation}\label{2309i111akfnkj}
\chi+\frac{1}{\sqrt{3}}(\psi\cdot v_l)(\psi\cdot \hat{v}_l)=\chi_0+\frac{1}{\sqrt{3}}(\psi_0\cdot v_l)(\psi_0\cdot \hat{v}_l),\quad\quad
v_l\cdot\psi=v_l\cdot\psi_0.
\end{equation}
In particular, the two maps respect the same rod data set. 
\end{theorem}

The first step towards establishing this result is to obtain a relationship between the distance $d(\Psi_0,\Psi)$ and the Mazur quantity $\mathrm{Tr}(\Psi_{0}^{-1}\Psi)$. Since the metric on $N$ is $G_{2(2)}$-invariant, the distance function is preserved under the action of left translation
\begin{equation}
d(\Psi_0,\Psi)=
d(Id,L_{\mathcal{B}^{-1}}\Psi),
\end{equation}
where $\mathcal{B}\in SL(7,\mathbb{R})$ arises from the representation of $G_{2(2)}$ and satisfies $\mathcal{B}\mathcal{B}^{T}=\Psi_0$ (see \cite[Section 5]{Clement:2007qy}), with
\begin{equation}
L_{\mathcal{B}^{-1}}\Psi=\mathcal{B}^{-1}
\Psi(\mathcal{B}^{-1})^{T}=e^{W}
\end{equation}
for some symmetric $W$ such that $\mathrm{Tr}\text{ }W=0$. Due to the fact that $N$ is a symmetric space, the Riemannian exponential map and the matrix exponential coincide. Furthermore, Hadamard's theorem applies (using that $N$ is complete, simply connected, and nonpositively curvatured) to show that the exponential map is a diffeomorphism, and the geodesic $\gamma(t)= e^{tW}$ is minimizing. It follows that
\begin{equation}
d(Id,L_{\mathcal{B}^{-1}}\Psi)=|\gamma'(0)|=|W|
=\sqrt{\mathrm{Tr}(W^2)}.
\end{equation}

Now consider the Mazur quantity \cite{Mazur:1982db}, namely
\begin{align}
\begin{split}
\mathrm{Tr}\left(\Psi_{0}^{-1}\Psi\right)=&
\mathrm{Tr}\left((\mathcal{B}^{-1})^{T}\mathcal{B}^{-1}
\Psi(\mathcal{B}^{-1})^{t}\mathcal{B}^{T}\right)\\
=&\mathrm{Tr}\left(\mathcal{B}^{-1}\Psi(\mathcal{B}^{-1})^T\right)\\
=&\mathrm{Tr}\text{ } e^{W}.
\end{split}
\end{align}
Since $e^{W}$ is symmetric and positive definite it may be diagonalized with positive eigenvalues
$\lambda_{i}$, $i=1,\ldots,7$. We then have
\begin{equation}
\mathrm{Tr}\text{ }e^{W}=\sum_{i=1}^{7}\lambda_{i},\quad\quad\quad
\mathrm{Tr}(W^2)=\sum_{i=1}^{7}(\log\lambda_{i})^2,
\end{equation}
and since $W$ has zero trace
\begin{equation}\label{alfjlkhna;tgyh}
\sum_{i=1}^{7}\log\lambda_i=0.
\end{equation}
If $\mathrm{Tr}\text{ }e^{W}\leq c$ then it is not difficult to see that \eqref{alfjlkhna;tgyh} implies $\mathrm{Tr}(W^2)\leq c_1$. Conversely if $\mathrm{Tr}(W^2)\leq c^2$
then each $|\log\lambda_{i}|\leq c$, and it holds that $\mathrm{Tr}\text{ }e^{W}\leq 7e^{c}$. We have thus shown the following.

\begin{lemma}\label{fflem1}
The distance $d(\Psi_0,\Psi)$ is uniformly bounded if and only if
the Mazur quantity $\mathrm{Tr}\left(\Psi_{0}^{-1}\Psi\right)$
is uniformly bounded.
\end{lemma}

\begin{proof}[Proof of Theorem \ref{aoifjikhjgsa}.]
If $\Psi$ is asymptotic to $\Psi_0$ then $d(\Psi_0,\Psi)\leq c_0$, that is the distance is uniformly bounded, in particular near $\Gamma$.
By Lemma \ref{fflem1} this implies that the Mazur function is also uniformly bounded
\begin{equation}
\mathrm{Tr}\left(\Psi_{0}^{-1}\Psi\right)\leq c.
\end{equation}
Moreover this quantity may be computed, as is done in Appendix \ref{appendixb} with all relevant notation, to find
\begin{align}\label{1100}
\begin{split}
\text{Tr}\left(\Psi^{-1}_0\Psi\right) =& \text{Tr}(F_0^{-1} F) + f^{-1} \text{Tr}( F_0^{-1} L_1^T L_1) +  \text{Tr}(F_0^{-1} L_5^T F^{-1} L_5) \\ 
&+ 2  \text{Tr}[F_0^{-1} (\nu - \nu_0)(\nu-\nu_0)^T] + f_0f^{-1} +  f^{-1} \text{Tr} [F_0 \mathbf{J} (\nu - \nu_0) (\mathbf{J}(\nu - \nu_0))^T ] \\  
&+ \text{Tr}(F_0 F^{-1})  + f_0^{-1} (\mathbf{J}(\nu - \nu_0))^T F \mathbf{J} (\nu - \nu_0) + f_0^{-1} f^{-1} L_2^2+ f_0^{-1} L_4^T F^{-1} L_4  \\ 
&+ f_0^{-1} f + 2f_0^{-1} \left(\mu - \mu_0 + \nu_0^T \mathbf{J} \nu\right)^2 + f^{-1} L_3^2+ 2 (\nu - \nu_0)^T F^{-1} (\nu - \nu_0) + 1.
\end{split}
\end{align} 
Since each of the terms on the right-hand side is nonnegative (see appendix), and the roles of $\Psi$ and $\Psi_0$ may be reversed, we have
\begin{equation}
c^{-1}f_{0}\leq f\leq cf_{0},\quad\quad\quad\quad
\mathrm{Tr}(FF_{0}^{-1})\leq c,
\end{equation}
\begin{equation}\label{fkjhi2uh}
(\nu - \nu_0)^T F_0^{-1} (\nu - \nu_0)\leq c/2,\quad\quad
f_0^{-1} \left(\mu - \mu_0 + \nu_0^T \mathbf{J} \nu\right)^2 \leq c/2,\quad\quad
f^{-1} L_4^T F_0^{-1} L_4\leq c.
\end{equation}

We now show that $F$ and $F_0$ give rise to the same rod structure.
Observe that since $F_0$ is symmetric it may be diagonalized with an orthogonal matrix $O$, so that $F_0=O\mathbf{D}O^T$ where $\mathbf{D}=\mathrm{diag}(\varrho_1, \varrho_2)$. Consider now a neighborhood of an interior point on an axis rod. At the axis points, the kernel of $F_0$ is 1-dimensional and so it may be assumed without loss of generality that $c_1^{-1}f_0\leq \varrho_1\leq c_1 f_0$ and $0<c_2^{-1}\leq \varrho_2\leq c_2$. Let $\tilde{F}=O^T FO$ then
\begin{equation}
\mathrm{Tr}(FF_{0}^{-1})=\mathrm{Tr}(FO \mathbf{D}^{-1}O^T)=\mathrm{Tr}(O O^T FO \mathbf{D}^{-1}O^T)
=\mathrm{Tr}(\tilde{F}\mathbf{D}^{-1})
=\tilde{f}_{11}\varrho_{1}^{-1}+\tilde{f}_{22}\varrho_{2}^{-1}.
\end{equation}
Thus
\begin{equation}
\tilde{f}_{11}\varrho_{2}+\tilde{f}_{22}\varrho_{1}\leq c\varrho_1 \varrho_2= cf_0,
\end{equation}
so that
\begin{equation}
\tilde{f}_{11}\leq c c_2 f_0,\quad\quad\quad \tilde{f}_{22}\leq cf_0\varrho_1^{-1}\leq cc_1.
\end{equation}
Moreover
\begin{equation}
f=\tilde{f}_{11}\tilde{f}_{22}-\tilde{f}_{12}^{2}\leq\tilde{f}_{11}\tilde{f}_{22}
\leq cc_2f_{0}\tilde{f}_{22},
\end{equation}
which produces the lower bound
\begin{equation}
\tilde{f}_{22}\geq (cc_2)^{-1}f f_0^{-1}\geq c^{-2}c_{2}^{-1}.
\end{equation}
In order to control the cross terms, observe that from the above
\begin{equation}
\tilde{f}_{12}^{2}=\tilde{f}_{11}\tilde{f}_{22}-f\leq c^2 c_1 c_2f_0.
\end{equation}
We conclude that
\begin{equation}
\tilde{f}_{11}\leq c_3 f_0,\quad\quad |\tilde{f}_{12}|\leq c_3 \sqrt{f_0},\quad\quad
c_3^{-1}\leq \tilde{f}_{22}\leq c_3.
\end{equation}
Hence, on an axis rod both $\mathbf{D}=O^T F_0 O$ and $\tilde{F}=O^T FO$ have the same kernel, and therefore $F_0$ and $F$ have the same kernel. Analogous arguments hold for a horizon rod.

Let us now show that the potential constants agree on axis rods. Consider an axis rod $\Gamma_l$, which we may assume without loss of generality has rod structure $v_l=(1,0)$, and use the notation above for diagonalizing $F_0$. The first inequality of \eqref{fkjhi2uh} yields
\begin{equation}
(\tilde{\psi}-\tilde{\psi}_0)^{T}\mathbf{D}^{-1}(\tilde{\psi}-\tilde{\psi}_0)
=3(\nu-\nu_0)^{T}F_{0}^{-1}(\nu-\nu_0)\leq 3c/2,
\end{equation}
where
\begin{equation}
(\tilde{\psi}-\tilde{\psi}_0)=O^T(\psi-\psi_0).
\end{equation}
It follows that
\begin{equation}\label{fjiuqi8hosihhga}
\varrho_{1}^{-1}(\tilde{\psi}_1-(\tilde{\psi}_{0})_1)^2
+\varrho_{2}^{-1}(\tilde{\psi}_2-(\tilde{\psi}_{0})_2)^2\leq 3c/2,
\end{equation}
which implies
\begin{equation}
(\tilde{\psi}_1-(\tilde{\psi}_{0})_1)^2\leq c_4 f_0.
\end{equation}
We then have $\psi_1=(\psi_0)_1$ on $\Gamma_l$, since $O$ coincides with the identity matrix on the axis. Next, observe that the second inequality of \eqref{fkjhi2uh} shows that on the axis
\begin{align}\label{====}
\begin{split}
0=&\sqrt{3}\left(\mu-\mu_0 +\nu_0^T \mathbf{J} \nu\right)\\
=&\chi-\chi_0 +\frac{1}{\sqrt{3}}\left(\psi_2 (\psi_0)_1 -\psi_1 (\psi_0)_2\right)\\
=& \chi+\frac{1}{\sqrt{3}}\psi_1 \psi_2 -\left(\chi_0 +\frac{1}{\sqrt{3}}(\psi_0)_1 (\psi_0)_2\right),
\end{split}
\end{align}
where in the last equality we used $\psi_1=(\psi_0)_1$ on $\Gamma_l$. Thus, \eqref{2309i111akfnkj} holds. To confirm \eqref{2309i11akhkkae} and \eqref{2309i11akhkkae1}, note that similar arguments to those that produced \eqref{fjiuqi8hosihhga} show that the last inequality of \eqref{fkjhi2uh} gives $L_4=0$ on the axis. A direct calculation of the components of $L_4$ then produces
\begin{equation}
0= \zeta_1 -(\zeta_0)_1+\frac{1}{3}(\chi-\chi_0)(2\psi_1 +(\psi_0)_1)-\frac{1}{3\sqrt{3}}\psi_1 (\psi_1 (\psi_0)_2 -\psi_2 (\psi_0)_1 ),   
\end{equation}
\begin{equation}
0= \zeta_2 -(\zeta_0)_2+\frac{1}{3}(\chi-\chi_0)(2\psi_2 +(\psi_0)_2)-\frac{1}{3\sqrt{3}}\psi_2 (\psi_1 (\psi_0)_2 -\psi_2 (\psi_0)_1 ). 
\end{equation}
Using \eqref{====} to replace $\chi-\chi_0$, as well as $\psi_1=(\psi_0)_1$ leads to
\begin{equation}
0= \zeta_1 -\frac{2}{3\sqrt{3}}\psi_1^2 \psi_2 -\left((\zeta_0)_1-\frac{2}{3\sqrt{3}}(\psi_0)_1^2 (\psi_0)_2 \right),
\end{equation}
\begin{equation}
0= \zeta_2 -\frac{1}{3\sqrt{3}}\psi_1 \psi_2^2 -\left((\zeta_0)_2-\frac{1}{3\sqrt{3}}(\psi_0)_1 (\psi_0)_2^2 \right),
\end{equation}
which yields the desired result.
\end{proof}

\section{Uniqueness of Minimal Supergravity Solutions}
\label{sec8}\setcounter{equation}{0}
\setcounter{section}{8}

In the previous section uniqueness was established for harmonic maps which are asymptotic to one another. This does not necessarily imply that any two minimal supergravity solutions having the same charges and rod structure are equivalent.
This is due to the fact that although two such solutions produce two harmonic maps $\Psi_1$ and $\Psi_2$, it is not known a priori that these maps remain within bounded distance to each other globally. Thus, the primary task of this section is
to show that indeed the distance $d(\Psi_1,\Psi_2)$ is uniformly bounded. Previous works \cite{Armas:2014gga,Armas:2009dd,Tomizawa:2010xj,Tomizawa:2009ua,Tomizawa:2009tb} on the uniqueness question for the minimal supergravity equations appear to have used the Mazur quantity $\mathrm{Tr}(\Psi_1 \Psi_{2}^{-1}-I)$, as opposed to the distance function. 
Both functions are subharmonic, and once they are known to be bounded a maximum principle argument may be used (as in Section \ref{sec7}) to yield that they vanish identically. A drawback to the prior approach is that in the minimal supergravity setting the Mazur quantity is difficult to compute, and so only special cases of uniqueness have been established previously. On the other hand these two subharmonic functions are related in that boundedness of one implies boundedness of the other.
This fact is a consequence of the structure of the symmetric space target, and may be proved as in \cite[Lemma 12]{KhuriWeinsteinYamada} which treats the vacuum case. The only difference here, when passing from vacuum to minimal supergravity, is the presence of extra potential terms which are treated in the same manner as the vacuum potentials in the proof of \cite[Lemma 12]{KhuriWeinsteinYamada}.

As in the vacuum case \cite{HollandsYazadjiev}, there are five regions to consider when establishing boundedness of the distance function. Namely: $1)$ the interior of axis rods, $2)$ the interior of horizon rods, $3)$ a neighborhood of infinity, $4)$ a neighborhood of the poles, which are intersections of a horizon and axis rod, and $5)$ a neighborhood of corner points, which are the intersection of two axis rods. It has been shown \cite[\S 2]{Armas:2014gga} that the Mazur quantity remains bounded in a neighborhood of infinity, and as mentioned above this implies boundedness of the distance function
in region $(3)$. Furthermore, the harmonic map does not blow-up at a horizon rod and thus the distance is controlled in region $(2)$. The arguments needed for regions $(4)$ and $(5)$ are similar to those of $(1)$, which we will treat first.

Consider an axis rod $\Gamma_l$ having rod structure $v_l$. By assumption both solutions have the same rod structure, so in particular this rod and its structure are shared. Since the linear combination $v_l^i \eta_{(i)}$ vanishes on this rod, the definitions \eqref{Thetadefinition}, \eqref{psi^i}, and \eqref{Upsilon} imply that $\Theta_i$, $d\left(v_l^i \psi_i\right)$, and $\Upsilon$ vanish on $\Gamma_l$. The computations \eqref{psiconstant}, \eqref{chiformula}, and \eqref{zetaformula} then show that there exist constants $a_l^{\mathbf{j}}$, $b_l^{\mathbf{j}}$, and $c_l^{\mathbf{j}}$ such that on this rod
\begin{equation}\label{2309i}
\zeta^{\mathbf{j}}\cdot v_l-\frac{2}{3\sqrt{3}}(\psi^{\mathbf{j}}\cdot v_l)^2(\psi^{\mathbf{j}}\cdot
\hat{v}_l)=a_l^{\mathbf{j}},\quad\quad
\zeta^{\mathbf{j}}\cdot \hat{v}_l-\frac{1}{3\sqrt{3}}(\psi^{\mathbf{j}}\cdot v_l)(\psi^{\mathbf{j}}\cdot\hat{v}_l)^2
=\hat{a}_l^{\mathbf{j}},
\end{equation}
\begin{equation}\label{2309i1}
\chi^{\mathbf{j}}+\frac{1}{\sqrt{3}}(\psi^{\mathbf{j}}\cdot v_l)(\psi^{\mathbf{j}}\cdot \hat{v}_l)=b_l^{\mathbf{j}},\quad\quad
v_l\cdot\psi^{\mathbf{j}}=c_l^{\mathbf{j}},
\end{equation}
where $\mathbf{j}=1,2$ indicates association with the solution $\Psi_{\mathbf{j}}$. Here, as before, it is assumed without loss of generality that $|v_l|=|\hat{v}_l|=1$.
It will be shown below that equality of angular momenta and charges of the two solutions implies that these constants agree on all axis rods, that is $a_l^{\mathbf{j}}=a_l$, $b_l^{\mathbf{j}}=b_l$, and $c_l^{\mathbf{j}}=c_l$. We then have
\begin{equation}\label{differenceequation1}
(\zeta^1-\zeta^2)\cdot v_l-\frac{2(c_l)^2}{3\sqrt{3}}(\psi^1-\psi^2)\cdot
\hat{v}_l=O(\rho^2),\quad
(\zeta^1-\zeta^2)\cdot \hat{v}_l-\frac{c_l}{3\sqrt{3}}\left[(\psi^1\cdot\hat{v}_l)^2
-(\psi^2\cdot\hat{v}_l)^2\right]
=O(\rho^2),
\end{equation}
\begin{equation}\label{differenceequation2}
(\chi^1-\chi^2)+\frac{c_l}{\sqrt{3}}(\psi^1-\psi^2)\cdot \hat{v}_l=O(\rho^2),\quad\quad
v_l\cdot(\psi^1-\psi^2)=O(\rho^2).
\end{equation}

Let $\Psi(s)$, $s\in[0,1]$ be a curve in the symmetric space target $\tilde{N}$ with $\Psi(0)=\Psi_2$ and $\Psi(1)=\Psi_1$. The distance is by definition the infimum of the length of all curves connecting the two solutions, and therefore $d(\Psi_1,\Psi_2)\leq L(\Psi(s))$. All the components of the curve, except for one, will be chosen to be linear functions. Namely
\begin{equation}
f_{ij}(s)=f_{ij}^2+s\left(f_{ij}^1 -f_{ij}^2\right),\quad
\chi(s)=\chi^2 +s\left(\chi^1-\chi^2\right), \quad
\psi(s)=\psi^2 +s\left(\psi^1 -\psi^2\right),
\end{equation}
\begin{equation}\label{plmk}
v_l\cdot\zeta(s)=v_l\cdot \zeta^2 +s v_l\cdot\left(\zeta^1-\zeta^2\right),
\quad
\hat{v}_l\cdot\zeta(s)=\hat{a}_l+\frac{c_l}{3\sqrt{3}}\left(\psi(s)\cdot\hat{v}_l\right)^2
+\gamma(s),
\end{equation}
where $\gamma(s)$ is a function satisfying
\begin{equation}
\gamma(0)=\hat{v}_l\cdot\zeta^2-\hat{a}_l-\frac{c_l}{3\sqrt{3}}
\left(\psi^2\cdot\hat{v}_l\right)^2,
\quad\quad
\gamma(1)=\hat{v}_l\cdot\zeta^1-\hat{a}_l-\frac{c_l}{3\sqrt{3}}
\left(\psi^1\cdot\hat{v}_l\right)^2.
\end{equation}
Observe that by \eqref{2309i} and \eqref{2309i1} both $\gamma(0)$ and $\gamma(1)$ are $O(\rho^2)$, and thus this function may be chosen so that $|\gamma(s)|+|\gamma'(s)|=O(\rho^2)$ for all $s$.

We will now estimate the length
\begin{equation}
L(\Psi(s))=\int_{0}^{1}\sqrt{G_{AB}\dot{\Psi}^A \dot{\Psi}^B}ds,
\end{equation}
where $G$ is the symmetric space metric given by \eqref{themetric} and $\dot{\Psi}=\partial_s \Psi$. The two terms of $G$ involving $dF$ remain uniformly bounded independent of $s$ since both solutions have the same rod structure, see the proof of \cite[Theorem 5]{HollandsYazadjiev}. In particular, observe that if $\lambda(s)$ and $\hat{\lambda}(s)$ are the eigenvalues of $F(s)$, then near $\Gamma_l$ we have the approximate diagonalization
\begin{equation}\label{diaggg}
F(s)=\lambda(s) v_l v_l^T+\hat{\lambda}(s)\hat{v}_l \hat{v}_l^T+O(\rho^2),\quad\quad
F^{-1}(s)=\lambda(s)^{-1} v_l v_l^T+\hat{\lambda}(s)^{-1}\hat{v}_l \hat{v}_l^T+O(\rho^2),
\end{equation}
where the eigenvalues are positive away from the axis with $\lambda(s)\sim\rho^2$ and $\hat{\lambda}(s)\sim 1$ away from corner points.
It follows that
\begin{equation}\label{FFFF}
F^{-1}\dot{F}=\lambda(s)^{-1}(\lambda^1-\lambda^2)v_l v_l^T
+\hat{\lambda}(s)^{-1}(\hat{\lambda}^1-\hat{\lambda}^2)\hat{v}_l \hat{v}_l^T+O(1)=O(1),
\end{equation}
showing that the first two terms of \eqref{themetric} possess the desired behavior. To proceed, write
\begin{equation}\label{first}
\psi=(\psi\cdot v_l)v_l+(\psi\cdot \hat{v}_l)\hat{v}_l,
\end{equation}
and use \eqref{differenceequation2} to find that the last term of \eqref{themetric} is controlled
\begin{equation}
\dot{\psi}^T F^{-1} \dot{\psi}
=\lambda^{-1}[(\psi^1-\psi^2)\cdot v_l]^2+\hat{\lambda}^{-1}[(\psi^1 -\psi^2)\cdot
\hat{v}_l]^2 +O(\rho^2) =O(1).
\end{equation}
Similar considerations show that
\begin{equation}
f^{-1}\Theta^T F^{-1}\Theta=f^{-1}\lambda^{-1}[\Theta\cdot v_l]^2+f^{-1}\hat{\lambda}^{-1}[\Theta\cdot
\hat{v}_l]^2 +O(\rho^2),
\end{equation}
where according to \eqref{nbnbnbm} and \eqref{differenceequation1}, \eqref{differenceequation2}
\begin{align}
\begin{split}
\Theta\cdot v_l=&
\partial_s\left[(\zeta\cdot v_l)-\frac{2}{3\sqrt{3}}(\psi\cdot v_l)^2(\psi\cdot
\hat{v}_l)\right]+O(\rho^2)\\
=&
(\zeta^1-\zeta^2)\cdot v_l-\frac{2(c_l)^2}{3\sqrt{3}}(\psi^1-\psi^2)\cdot\hat{v}_l
+O(\rho^2)\\
=&O(\rho^2),
\end{split}
\end{align}
and with the help of \eqref{plmk}
\begin{align}
\begin{split}
\Theta\cdot \hat{v}_l=&
\partial_s\left[(\zeta\cdot \hat{v}_l)-\frac{1}{3\sqrt{3}}(\psi\cdot v_l)(\psi\cdot\hat{v}_l)^2\right]+O(\rho^2)\\
=&
O(\rho^2).
\end{split}
\end{align}
In light of the fact that $f(s)\sim \rho^2$, we then have bounds for the third term of \eqref{themetric}, that is
\begin{equation}\label{last}
f^{-1}\Theta^T F^{-1}\Theta=O(1).
\end{equation}
Analogous arguments yield $f^{-1}\Upsilon^2=O(1)$, and consequently $L(\Psi(s))=O(1)$. Therefore, the distance $d(\Psi_1,\Psi_2)$ is bounded in a neighborhood of the interior of axis rods.

Consider now region (4) consisting of a neighborhood of the poles, which are intersections of a horizon rod and an axis rod $\Gamma_l$. It may be assumed for the purposes of this argument that the pole in question lies at the origin in the $\rho z$-plane. Here we will follow closely the arguments of Hollands and Yazadjiev \cite[pgs. 668-69]{HollandsYazadjiev}, who treated the vacuum case. By redefining the torus fiber coordinates $(\phi^1,\phi^2)$ if necessary, we may assume without loss of generality that the rod structure $v_{l}^T =(1,0)$; we then also have $\hat{v}_l^T=(0,1)$. As explained in \cite[pg. 668]{HollandsYazadjiev}, associated with each solution are a set of coordinates $(R_1,Y_1)$ and $(R_2,Y_2)$ for the orbit space $\hat{M}^2=M^5/[U(1)^2 \times \mathbb{R}]$ near the pole such that
\begin{equation}
 F^{\mathbf{j}}=
  \left( {\begin{array}{cc}
   R_{\mathbf{j}}^2(1+O(R_{\mathbf{j}}^2)) & R_{\mathbf{j}}^2 O(1) \\
    R_{\mathbf{j}}^2 O(1) & \mathbf{d} +O(R_{\mathbf{j}}^2+Y_{\mathbf{j}}^2) \\
  \end{array} } \right),
\end{equation}
and
\begin{equation}\label{8900}
\zeta^{\mathbf{j}}\cdot v_l-\frac{2}{3\sqrt{3}}(\psi^{\mathbf{j}}\cdot v_l)^2(\psi^{\mathbf{j}}\cdot
\hat{v}_l)=a_l+O(R_{\mathbf{j}}^2),\quad\quad
\zeta^{\mathbf{j}}\cdot \hat{v}_l-\frac{1}{3\sqrt{3}}(\psi^{\mathbf{j}}\cdot v_l)(\psi^{\mathbf{j}}\cdot\hat{v}_l)^2
=\hat{a}_l+O(R_{\mathbf{j}}^2),
\end{equation}
\begin{equation}\label{89001}
\chi^{\mathbf{j}}+\frac{1}{\sqrt{3}}(\psi^{\mathbf{j}}\cdot v_l)(\psi^{\mathbf{j}}\cdot \hat{v}_l)=b_l +O(R_{\mathbf{j}}^2),\quad\quad
v_l\cdot\psi^{\mathbf{j}}=c_l +O(R_{\mathbf{j}}^2),
\end{equation}
for $\mathbf{j}=1,2$ where $\mathbf{d}$ is a positive number. The new coordinates satisfy the properties that $R_{\mathbf{j}}\geq 0$, $R_{\mathbf{j}}=0$ corresponds to the axis $\Gamma_l$, and $R_{\mathbf{j}}(0)=Y_{\mathbf{j}}(0)=0$. In analogy with \eqref{diaggg},
near the pole we then have
\begin{equation}
F(s)=\lambda(s) v_l v_l^T+\hat{\lambda}(s)\hat{v}_l \hat{v}_l^T+O(R_1^2+R_2^2),\quad\quad
F^{-1}(s)=\lambda(s)^{-1} v_l v_l^T+\hat{\lambda}(s)^{-1}\hat{v}_l \hat{v}_l^T+O(1),
\end{equation}
where
\begin{equation}
\lambda(s)\sim s R_1^2 +(1-s) R_2^2,\quad\quad\quad\quad \hat{\lambda}(s)\sim \mathbf{d}.
\end{equation}
It is shown in \cite[pg. 669]{HollandsYazadjiev} that the two sets of coordinate functions $(R_{\mathbf{j}},Y_{\mathbf{j}})$, $\mathbf{j}=1,2$ are asymptotic to one another, and therefore as in \eqref{FFFF} we find that $F^{-1}(s)\dot{F}(s)=O(1)$. This ensures that the first two terms of \eqref{themetric} are appropriately controlled near poles. Furthermore,
in light of \eqref{8900}, \eqref{89001} the difference equations
\eqref{differenceequation1}, \eqref{differenceequation2} remain valid here with $O(\rho^2)$ replaced by $O(R_1^2 +R_2^2)$. It follows that we may imitate the arguments of \eqref{first}-\eqref{last} to establish boundedness of the remaining terms of \eqref{themetric}, so that $L(\Psi(s))=O(1)$. Hence, the distance $d(\Psi_1,\Psi_2)$ is bounded in a neighborhood of poles.

It remains to consider region (5), consisting of a neighborhood of the corner points where two axis rods intersect. This, however, may be treated in an analogous way to region (4). Namely, following
\cite[pg. 669]{HollandsYazadjiev}, new coordinates may be introduced in a neighborhood of a corner point, which elucidate the degeneracy present in the matrices $F^{\mathbf{j}}$. From there, as in case (4), the procedure given in case (1) may be employed to conclude that $L(\Psi(s))=O(1)$. Therefore, the distance $d(\Psi_1,\Psi_2)$ is bounded globally.

\begin{figure}[h]
	\tikzset{middlearrow/.style={
			decoration={markings,
				mark= at position 0.25 with {\arrow{#1}} ,
			},
			postaction={decorate}
		}
	}	
	\begin{tikzpicture}[scale=.6, every node/.style={scale=0.8}]
	\draw[<-](-12,0)--(-4.6,0);
	\draw[](-3.5,0)--(1,0);
	 \draw[](2.5,0)--(3.3,0)node[black,font=\large,below=.1cm]{$\Gamma_{l_4}$}--(4,0);
	\draw[->](5.5,0)--(11,0)node[black,font=\large,right=.1cm]{$z-$axis};
	\fill[black] (-9,0) circle [radius=.1]node[black,font=\large,below=.1cm]{$z_1$};
	\fill[black] (-6.5,0) circle [radius=.12];
	\node at (-6.2,-.5){$z_2 \!=\!z_{l_1-1}$};
	\fill[black] (-4.6,0) circle [radius=.12];
	\fill[black] (-3.5,0) circle [radius=.12];
	\fill[black] (-2,0) circle [radius=.12];
	\fill[black] (-.5,0) circle [radius=.12] node[black,font=\large,below=.1cm]{$z_{l_3}$};
	\fill[black] (1,0) circle [radius=.12];
	\fill[black] (2.5,0) circle [radius=.12];
	\fill[black] (4,0) circle [radius=.12];
	\fill[black] (5.5,0) circle [radius=.12];
	\draw[black,middlearrow={<}] (-6.5,0) arc (0:180:1.25)node[black,font=\large,midway,above=.1cm]{$v_1$} ;
	\draw[black,middlearrow={<}] (-2,0) arc (0:180:2.25)node[black,font=\large,midway,above=.1cm]{$v_1$} ;
	\draw[black,middlearrow={<}] (-.5,0) arc (0:180:.75)node[black,font=\large,midway,above=.1cm]{$v_1$} ;
	\draw[black,middlearrow={<}] (3,0) arc (0:180:1.75)node[black,font=\large,midway,above=.1cm]{$v_{l_4}$} ;
	\draw[black,middlearrow={>}] (7,0) arc (0:180:3.75)node[black,font=\large,midway,above=.1cm]{$v_{L+1}$} ;
	\draw[black,middlearrow={>}] (7,0) arc (0:180:6.75)node[black,font=\large,midway,above=.1cm]{$v_{L+1}$} ;
	\end{tikzpicture}
	\caption{Rod Diagram and Dipole Charges}\label{figure1}
\end{figure}
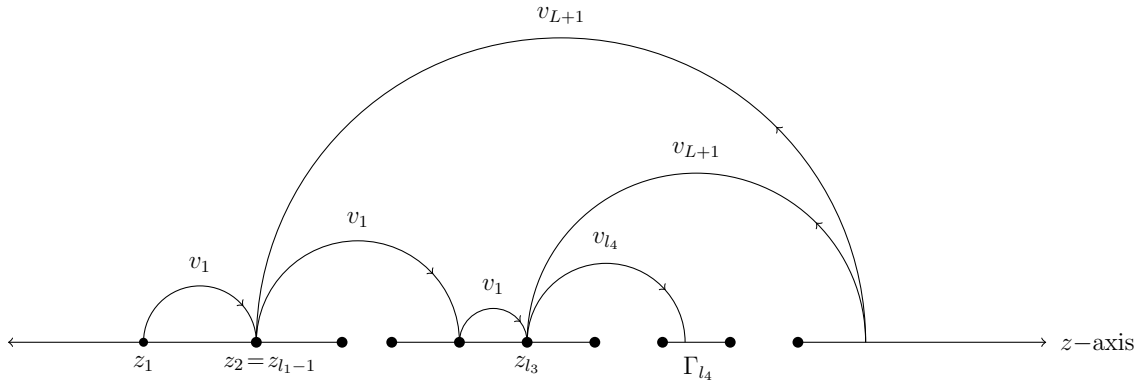

It remains to show that the constants $a_l^{\mathbf{j}}$, $b_l^{\mathbf{j}}$, and $c_l^{\mathbf{j}}$ are independent of $\mathbf{j}$. Let $l=1,\ldots,L+1$ enumerate the entire sequence of rods along the $z$-axis as in \eqref{rods}. There are $m$ horizon rods and $n$ axis rods.
We begin by showing that $c_l^1=c_l^2$ by demonstrating that these constants are uniquely determined by knowledge of $n-2$ dipole charges. Recall that these constants determine the magnetic potential in the
rod structure direction along an axis rod $\Gamma_l=[z_l,z_{l-1}]$, namely $v_l\cdot\psi^{\mathbf{j}}=c_l^{\mathbf{j}}$. Since the potential $\psi$ has two components it is defined up to the choice of two integration constants, which we choose to obtain $c_1^{\mathbf{j}}=c_{L+1}^{\mathbf{j}}=0$. Suppose first that all rod structures are pairwise linearly independent with $v_1$. Take a semi-circle in the $\rho z$-half plane orbit space emanating from the left most corner point $z_1$ to the next corner point $z_2$, see Figure \ref{figure1}. Together with the orbit $S^1$ associated with
$v_1^i\partial_{\phi^i}$, this semi-circle represents an $S^2$ bubble having dipole charge
\begin{equation}
\mathcal{D}_2(v_1)=v_1\cdot\left(\psi^{\mathbf{j}}(z_2)-\psi^{\mathbf{j}}(z_1)\right).
\end{equation}
Therefore, knowledge of $\mathcal{D}_2(v_1)$ and $v_1\cdot\psi^{\mathbf{j}}(z_1)=c_1^{\mathbf{j}}=0$ gives rise to knowledge of $v_1\cdot\psi^{\mathbf{j}}(z_2)$. Next, take a semi-circle connecting $z_2$ to the next corner point $z_3$. Then in the same way, knowledge of $\mathcal{D}_3(v_1)$ and $v_1\cdot\psi^{\mathbf{j}}(z_2)$ yields knowledge of $v_1\cdot\psi^{\mathbf{j}}(z_3)$. Continue this process down the $z$-axis until reaching a horizon rod or the last corner point. We then have determined $v_1\cdot\psi^{\mathbf{j}}(z_l)$, $l=1,\ldots, l_1 -1$ where
$\Gamma_{l_1 +1}=[z_{l_1 +1},z_{l_1}]$ is the first horizon rod. Now extend a semi-circle emanating from $z_{l_1 -1}$ to the next corner point $z_{l_2}$; note that this may require jumping over more than one horizon rod. This semi-circle has a $v_1$-dipole charge associated with it, and its value together with $v_1\cdot\psi^{\mathbf{j}}(z_{l_1 -1})$ determines $v_1\cdot\psi^{\mathbf{j}}(z_{l_2})$. We may proceed in this way, down to the last corner point, to obtain the $v_1$-direction of $\psi$ at all corner points from knowledge of $v_1$-dipole charges.

The next step involves a similar process going from the right of the $z$-axis leftwards. Consider the semi-infinite rod $\Gamma_{L+1}$ at the right of the $z$-axis. If this is part of a larger connected sequence of axis rods, then $z_L$ is a corner point. By the set up $v_{L+1}\cdot\psi^{\mathbf{j}}(z_{L})=c_{L+1}^{\mathbf{j}}=0$, and by the above process the value of $v_1\cdot\psi^{\mathbf{j}}(z_{L})$ is determined. Thus, since $v_1$ and $v_{L+1}$ are linearly independent, we know the whole vector $\psi^{\mathbf{j}}(z_L)$. It follows that $c_{L}^{\mathbf{j}}=v_{L}\cdot \psi^{\mathbf{j}}$ is determined on $\Gamma_L$. Since $v_{L}$ is linearly independent with $v_1$, if $z_{L-1}$ is another corner point we may similarly
determine $c_{L-1}^{\mathbf{j}}$. In fact, this may be continued to obtain all $c_{l}^{\mathbf{j}}$ associated with this connected sequence of axis rods. Take now
a semi-circle emanating from  the bottom axis rod $\Gamma_{L+1}$ and ending on
a corner point $z_{l_3}$ of another connected sequence of axis rods. Knowledge of the $v_{L+1}$-dipole charge affiliated with this semi-circle then gives knowledge
of $v_{L+1}\cdot\psi^{\mathbf{j}}(z_{l_3})$, from which we may determine the whole vector $\psi^{\mathbf{j}}(z_{l_3})$ since $v_{1}\cdot\psi^{\mathbf{j}}(z_{l_3})$ is already known. As before this yields all constants $c_{l}^{\mathbf{j}}$ inherent to this sequence.

The above process uniquely determines all constants $c_{l}^{\mathbf{j}}$ from knowledge of dipole charges, except those arising from axis rods which are bordered
by two horizon rods. Consider such an axis rod $\Gamma_{l_4}$, and take a semi-circle connecting it to a corner point $z_{l_5}$ (in Figure \ref{figure1}, $z_{l_5}=z_{l_3}$). Then knowledge of the $v_{l_4}$-dipole charge coming from this semi-circle determines $c_{l_4}^{\mathbf{j}}$, since the whole vector $\psi^{\mathbf{j}}(z_{l_5})$ has previously been determined. It should be noted that the corner point $z_{l_5}$ may be the `corner at infinity' if no proper corner points are present. Moreover, this algorithm was carried out with the initial assumption that all rod structures are pairwise linearly independent with $v_1$. However, straightforward modifications may be made in the case that this assumption is not valid. Finally, simple bookkeeping shows that a total of $n-2$ dipole charges are used to fully determine the constants and show that $c_{l}^{1}=c_{l}^{2}$.

We now treat the constants $b_{l}^{\mathbf{j}}$ and show that they are uniquely
determined by dipole charges and the electric charge of each horizon component. As above we choose a gauge in which $c_1^{\mathbf{j}}=c_{L+1}^{\mathbf{j}}=0$. This implies, by \eqref{chiformula}, that $\chi^{\mathbf{j}}(\Gamma_1)=b_1^{\mathbf{j}}$
and $\chi^{\mathbf{j}}(\Gamma_{L+1})=b_{L+1}^{\mathbf{j}}$. The potential $\chi$ is defined up to a constant, and by an appropriate choice of this constant we obtain
$\chi^{\mathbf{j}}(\Gamma_1)=b_1^{\mathbf{j}}=0$. Since the total electric charge (evaluated at infinity) is expressed in terms of the difference $\chi^{\mathbf{j}}(\Gamma_1)-\chi^{\mathbf{j}}(\Gamma_{L+1})$, which in turn is given as a sum of dipole charges and electric charges of horizon components, the value $b_{L+1}^{\mathbf{j}}$ is determined. We now proceed from top to bottom in a step by step fashion along the $z$-axis. Observe that from \eqref{chiformula}
\begin{equation}
b_2^{\mathbf{j}}
=\chi^{\mathbf{j}}(z_1)+\frac{c_{2}}{\sqrt{3}}\hat{v}_2\cdot\psi^{\mathbf{j}}(z_1)
=\frac{c_{2}}{\sqrt{3}}\hat{v}_2\cdot\psi^{\mathbf{j}}(z_1),
\end{equation}
and the right-hand side is fully determined by dipole charges. With $b_2^{\mathbf{j}}$ in hand we then know the value of
\begin{equation}
\chi^{\mathbf{j}}(z_2)=b_2^{\mathbf{j}}
-\frac{c_{2}}{\sqrt{3}}\hat{v}_2\cdot\psi^{\mathbf{j}}(z_2)
\end{equation}
in terms of dipole charges. Continuing in this way down the axis, the values of $b_l^{\mathbf{j}}$ and the values of $\chi^{\mathbf{j}}$ at corner points before the first horizon rod $\Gamma_{l_1 +1}$ are known. Since the charge of this horizon component is given in terms of the difference $b_{l_1 +2}^{\mathbf{j}}-b_{l_1}^{\mathbf{j}}$, it follows that $b_{l_1 +2}^{\mathbf{j}}$ is determined. This process may now be repeated until all the constants are found in terms of the fixed dipole and electric charges. We then have $b_{l}^1=b_{l}^2$.

Finally the angular momentum constants may be treated analogously to those of electric charge, so that they are uniquely determined by the angular momenta of the horizon components and dipole charges resulting in $a_{l}^{1}=a_{l}^{2}$.

\appendix
\section{Conical Singularities and Geometric Regularity}
\label{sec9}
\renewcommand{\theequation}{A.\arabic{equation}}
\setcounter{equation}{0}

There are two possible regularity issues that can arise when constructing the spacetime $(M^5,\mathbf{g})$ with metric
\begin{equation}\label{fjhjhgiikri}
\mathbf{g}=-f^{-1}\rho^2 dt^2+e^{2\alpha}(d\rho^2+dz^2)
+f_{ij}(d\phi^{i}+\omega^{i}dt)(d\phi^{j}+\omega^{j}dt),\quad\quad \alpha =\sigma-\tfrac{1}{2}\log f,
\end{equation}
from the harmonic map $\varphi:\mathbb{R}^3 \setminus\Gamma\rightarrow G_{2(2)}/SO(4)$. More precisely, these are the questions of \textit{geometric regularity} and \textit{analytic regularity}. Geometric regularity concerns the ability to smoothly extend the spacetime metric across the rods and is related to the
potential presence of conical singularities, while analytic regularity concerns the differentiability properties of the harmonic map up to the orbit space boundary after the singular part has been removed.  We will not treat the issue of analytic regularity here, and note that it has only relatively recently been established for the Einstein-Maxwell equations in the classical 4D setting by Nguyen \cite{N}, for the interior of axis rods. In this setting, the question of analytic regularity at poles is apparently still open. On the other hand, the 4D vacuum case was treated independently by Li-Tian \cite{LT,LT93} and Weinstein \cite{Wei92}. In this appendix we will show that, assuming a minimal amount of analytic regularity for the harmonic maps constructed above, the question of geometric regularity is resolved precisely when the axes are devoid of conical singularities. 

To begin, we will examine geometric regularity at the interior of a horizon rod $\Gamma_h$. Note that since the harmonic map remains bounded at such points, equation \eqref{fjhkjhhw} shows that  $\Omega^i=-\omega^i|_{\Gamma_h}$, $i=1,2$ are constant on the horizon and therefore the Killing field $\partial_t+\Omega^i \partial_{\phi^i}$ becomes null there, making each such rod represent a Killing horizon. By using the change of coordinates $\tilde{t}=t$, $\tilde{\phi}^i=\phi^i -\Omega^i t$ (where $\phi^i$ are viewed as coordinates on the universal cover $\mathbb{R}^2$ of $T^2$), the metric becomes
\begin{equation}
\mathbf{g}=-f^{-1}\rho^2 d\tilde{t}^2+e^{2\alpha}(d\rho^2+dz^2)
+f_{ij}(d\tilde{\phi}^{i}+\tilde{\omega}^{i}d\tilde{t})(d\tilde{\phi}^{j}+\tilde{\omega}^{j}d\tilde{t}),\quad\quad \alpha =\sigma-\tfrac{1}{2}\log f,
\end{equation}
with $\tilde{\omega}^i=\omega^i+\Omega^i=\rho^2 \bar{\omega}^i$ for some regular $\bar{\omega}^i$.
Furthermore, observe that from \eqref{PDEsigma} and \eqref{themetric} the following quadrature equations hold globally
\begin{align}\label{alphaeq}
\begin{split}
\alpha_\rho =&
\frac\rho8 \left[ (\log f)_\rho^2 - (\log f)_z^2 + \tr F^{-1}F_\rho F^{-1}F_\rho
- \tr F^{-1}F_z F^{-1}F_z
- 4\rho^{-1} (\log f)_\rho \right.\\
&\left.
+ 2f^{-1} \Theta^{T}_{\rho}F^{-1}\Theta_{\rho} -2f^{-1} \Theta^{T}_{z}F^{-1}\Theta_{z}
+2f^{-1}(\Upsilon^2_\rho -\Upsilon^2_z)+2\psi_{\rho}^{T} F^{-1}\psi_{\rho}-2\psi_{z}^{T} F^{-1} \psi_{z}
\right], \\
\alpha_z =& \frac\rho4 \left[ (\log f)_\rho(\log f)_z + \tr F^{-1}F_\rho F^{-1}F_z -2\rho^{-1} (\log f)_z  + 2f^{-1} \Theta^{T}_{\rho}F^{-1}\Theta_{z} \right.\\
&\left.
+2f^{-1}\Upsilon_\rho \Upsilon_z +2\psi_{\rho}^{T} F^{-1}\psi_{z}\right].
\end{split}
\end{align}
Thus, integrating \eqref{alphaeq} produces $\alpha=-\frac{1}{2}\log f-\log c +\bar{\alpha}$ near $\Gamma_h$ for a constant $c>0$ which is related to the horizon surface gravity \cite[Appendix]{HollandsYazadjiev}, and a function $\bar{\alpha}=O(\rho^2)$. It follows that
\begin{equation}
- f^{-1}\rho^2 d\tilde{t}^2 +e^{2\alpha}(d\rho^2 +dz^2)=c^{-2}f^{-1}e^{2\bar{\alpha}}\left(-c^2\rho^2 d\tilde{t}^2 +d\rho^2 +dz^2\right)+O(\rho^4) d\tilde{t}^2.
\end{equation}
Introducing Kruskal-type coordinates $X,Y>0$ given by $XY=\rho^2$, $X/Y =e^{2c\tilde{t}}$ then yields
\begin{align}
\begin{split}
\mathbf{g}=&c^{-2}f^{-1}e^{2\bar{\alpha}}\left(dX dY +dz^2 \right)+f_{ij}\left(d\tilde{\phi}^{i}+
\tfrac{\bar{\omega}^{i}}{2c}(YdX-XdY)\right)
\left(d\tilde{\phi}^{j}+\tfrac{\bar{\omega}^{i}}{2c}(YdX-XdY)\right)\\
&+O(1)(YdX-XdY)^2 ,
\end{split}
\end{align}
near the interior of $\Gamma_h$. Therefore, the geometry is regular across the interior of horizon rods without the need to balance certain parameters. 

Regularity of the solution at axis rods, corners, and poles, however, does rely on the balancing of parameters to relieve geometric singularities.  Consider first a small neighborhood $\mathcal{V}_a\subset \hat{M}^2$ in the orbit space of an interior point to an axis rod $\Gamma_l$, which does not intersect the endpoints of $\Gamma_l$. We may assume without loss of generality that the rod structure of $\Gamma_l$ is $(1,0)$. Assuming analytic regularity yields expansions of the harmonic map variables in $\rho$. In particular
\begin{equation}
F=
\left( {\begin{array}{cc}
e^u & O(\rho^2) \\
O(\rho^2) & e^{v} \\
\end{array} } \right)
\end{equation}
where $u=2\log\rho +\bar{u}$ in which $\bar{u}$ and $v$ are regular function, and as in \eqref{2309i}, \eqref{2309i1} the potentials satisfy
\begin{equation}\label{2309ij}
\zeta_1 +\frac{1}{3\sqrt{3}}\psi_{1}^2 \psi_2 +\psi_1 (\chi-b_l)=a_l +O(\rho^4),\quad\quad\quad
\zeta_2 -\frac{1}{3\sqrt{3}}\psi_1 \psi_2^2 
=\hat{a}_l +O(\rho^2),
\end{equation}
\begin{equation}\label{2309i1j}
\chi+\frac{1}{\sqrt{3}}\psi_1 \psi_2=b_l +O(\rho^2),\quad\quad\quad
\psi_1=c_l +O(\rho^2),
\end{equation}
for some potential constants $a_l$, $\hat{a}_l$, $b_l$, and $c_l$. This implies that the potential expressions within the brackets of \eqref{alphaeq} are bounded. It follows that
\begin{equation}
\alpha_{\rho}=\frac{1}{2}\bar{u}_{\rho}+O(\rho),\quad\quad\quad
\alpha_z =\frac{1}{2}\bar{u}_{z}+O(\rho^2),
\end{equation}
and hence $\alpha=\frac{1}{2}\bar{u}+c+O(\rho^2)$ in $\mathcal{V}_a$, for some constant $c$. The spacetime metric over $\mathcal{V}_a$ may then be written as
\begin{equation}\label{alofkihyaf}
\mathbf{g}=-e^{O(1)}dt^2 +e^{\bar{u}+2c+O(\rho^2)}(d\rho^2 +dz^2) +e^{u} (d\phi^1)^2 
+O(\rho^2)d\phi^1 d\phi^2 +e^{v} (d\phi^2)^2.
\end{equation}
Clearly, the absence of a conical singularity on $\Gamma_l$ is equivalent to $c=0$. In this case, we may make the change of coordinates $x=\rho \cos\phi^1$, $y=\rho \sin\phi^1$ to find
\begin{align}
\begin{split}
e^{\bar{u}+O(\rho^2)}d\rho^2 +e^{u} (d\phi^1)^2
=&e^{\bar{u}}\left(dx^2 +dy^2\right)
+O(1)\left(xdx+ydy\right)^2,\\
O(\rho^2)d\phi^1 d\phi^2=&O(1)\left(xdy-ydx\right)d\phi^2.
\end{split}
\end{align}
Therefore, the metric \eqref{alofkihyaf} is geometrically regular across the interior of the axis rod $\Gamma_l$.

Next let $\mathcal{V}_c\subset\hat{M}^2$ be a neighborhood of a corner, which separates two axis rods $\Gamma_1$ to the north and $\Gamma_2$ to the south. By performing a coordinate change in the torus fibers if necessary, we may assume without loss of generality that $\Gamma_1$, $\Gamma_2$ have rod structures  $(1,0)$, $(0,1)$ respectively. In addition, the origin of the orbit space coordinates may be taken to be the corner point intersection of these two axis rods, and the Euclidean distance to the origin will be denoted $r=\sqrt{\rho^2 +z^2}$. Then assuming analytic regularity and a sufficiently small $\mathcal{V}_c$, there exist regular functions $\bar{u}$, $\bar{v}$ such that $u=\log(r-z)+\bar u$, $v=\log(r+z)+\bar{v}$ contribute to the expansion
\begin{equation}
F=
\left( {\begin{array}{cc}
e^u & O(\rho^2) \\
O(\rho^2) & e^{v} \\
\end{array} } \right),
\end{equation}
and as in \eqref{2309i}, \eqref{2309i1} the potentials satisfy
\begin{equation}\label{2309ij1}
\zeta_1 +\frac{1}{3\sqrt{3}}\psi_{1}^2 \psi_2 +\psi_1 (\chi-b)=a +O(|r-z|^2),\quad\quad\quad
\zeta_2 -\frac{1}{3\sqrt{3}}\psi_1 \psi_2^2 
=\hat{a} +O(|r-z|),
\end{equation}
\begin{equation}\label{2309ij1'}
\zeta_2 +\frac{1}{3\sqrt{3}}\psi_{1} \psi_2^2 +\psi_2 (\chi-b)=\hat{a}-\frac{c_1 c_2^2}{\sqrt{3}} +O(|r+z|^2),\quad\quad\quad
\zeta_1 -\frac{1}{3\sqrt{3}}\psi_1^2 \psi_2 
=a+\frac{c_1^2 c_2}{\sqrt{3}}+O(|r+z|),
\end{equation}
\begin{equation}\label{2309i1j1}
\chi+\frac{1}{\sqrt{3}}\psi_1 \psi_2=b +O(\rho^2),\quad\quad
\psi_1=c_1 +O(|r-z|),\quad\quad \psi_2=c_2 +O(|r+z|),
\end{equation}
for some potential constants $a$, $\hat{a}$, $b$, $c_1$, and $c_2$. This implies that the potential expressions within the brackets of \eqref{alphaeq} are bounded. It follows that
\begin{align}
\begin{split}
\alpha_{\rho}=&-\frac{\rho}{2r^2}+\frac{(r+z)}{4r}\bar{u}_{\rho}
+\frac{(r-z)}{4r}\bar{v}_{\rho}+\frac{\rho}{4r}(\bar{u}_z -\bar{v}_z)+O(\rho),\\
\alpha_z =&-\frac{z}{2r^2}+\frac{(r+z)}{4r}\bar{u}_{z}
+\frac{(r-z)}{4r}\bar{v}_{z}-\frac{\rho}{4r}(\bar{u}_{\rho} -\bar{v}_{\rho})+O(\rho^2).
\end{split}
\end{align}
On $\mathcal{V}_c$ we then have
\begin{equation}
\alpha=-\frac{1}{2}\log r +\frac{(r+z)}{4r}(\bar{u}-\bar{u}(0))+\frac{(r-z)}{4r}(\bar{v}-\bar{v}(0))+c +O(\rho^2),
\end{equation}
for some constant $c$ and where $\bar{u}(0)$, $\bar{v}(0)$ denote these functions evaluated at the corner.  The spacetime metric over $\mathcal{V}_a$ may then be written as
\begin{align}
\begin{split}
\mathbf{g}=&-e^{O(1)}dt^2 +(r-z) e^{\bar{u}} (d\phi^1)^2 +(r+z)e^{\bar{v}} (d\phi^2)^2 +O(\rho^2)d\phi^1 d\phi^2\\
&+r^{-1}\exp\left(2c+\frac{(r+z)}{2r}(\bar{u}-\bar{u}(0))
+\frac{(r-z)}{2r}(\bar{v}-\bar{v}(0))+O(\rho^2)\right)(d\rho^2 +dz^2).
\end{split}
\end{align}
The absence of conical singularities on the two axis rods $\Gamma_1$, $\Gamma_2$ implies that $\bar{u}(0)=\bar{v}(0)$ and $2 e^{2c}=e^{\bar{u}(0)}=e^{\bar{v}(0)}$. Furthermore, note 
that the geometric angle at the pole between the rods is $\pi/2$, and not $\pi$ as it is in the orbit space. This motivates the change to new coordinates $\xi,\eta\geq 0$ given by $z+i\rho=\frac{1}{2}(\xi+i\eta)^2$, or equivalently $\rho=\xi\eta$, $z=\frac{1}{2}\left(\xi^2-\eta^2\right)$, in which the metric takes the form
\begin{align}
\begin{split}
\mathbf{g}=&-e^{O(1)} dt^2 +\eta^2 e^{\bar{u}} (d\phi^1)^2 +\xi^2 e^{\bar{v}} (d\phi^2)^2
+O(\xi^2 \eta^2)d\phi^1 d\phi^2\\
&+ \exp\left(\bar{u}(0)+\frac{\xi^2}{\xi^2+\eta^2}(\bar{u}-\bar{u}(0))+
\frac{\eta^2}{\xi^2 +\eta^2}(\bar{v}-\bar{v}(0))+O\left(\xi^2 \eta^2\right)\right)(d\xi^2 +d\eta^2).
\end{split}
\end{align}
Now define two pairs of Cartesian coordinates $x_1=\eta \cos\phi^1$, $y_1=\eta \sin\phi^1$, and $x_2=\xi \cos\phi^2$, $y_2=\xi \sin\phi^2$ with
\begin{equation}
d\eta^2 +\eta^2 (d\phi^1)^2=dx_1^2 +dy_1^2, \quad\quad \eta^2 d\eta^2=(x_1 dx_1+y_1dy_1)^2,
\end{equation}
\begin{equation}
d\xi^2 +\xi^2 (d\phi^2)^2=dx_2^2 +dy_2^2, \quad\quad \xi^2 d\xi^2=(x_2 dx_2+y_2dy_2)^2.
\end{equation}
We then have
\begin{align}
\begin{split}
\mathbf{g}=& -e^{O(1)}dt^2 +e^{\bar{u}}\left(\eta^2 (d\phi^1)^2 +d\eta^2
+O\left(\frac{\bar{u}-\bar{v}-\bar{u}(0)+\bar{v}(0)}{\xi^2 +\eta^2}+\xi^2\right) \eta^2 d\eta^2\right)\\
&+e^{\bar{v}}\left(\xi^2 (d\phi^2)^2 +d\xi^2+ O\left(\frac{\bar{u}-\bar{v}-\bar{u}(0)+\bar{v}(0)}{\xi^2 +\eta^2}+\eta^2\right) \xi^2 d\xi^2\right)+O(\xi^2 \eta^2)d\phi^1 d\phi^2\\
=&-e^{O(1)}dt^2 +e^{\bar{u}}\left(dx_1^2 +dy_1^2 +O\left(\bar{r}^2\right)(x_1 dx_1 +y_1 dy_1)^2 \right)\\
&+e^{\bar{v}}\left(dx_2^2 +dy_2^2 +O\left(\bar{r}^2\right)(x_2 dx_2 +y_2 dy_2)^2 \right)
+O(1)\left(y_1 dx_1 -x_1 dy_1\right)\left(y_2 dx_2 -x_2 dy_2\right),
\end{split}
\end{align}
where $\bar{r}^2=\sum (x_i^2 +y_i^2)$ and the regularity of $\bar{u}-\bar{v}$ at the corner has been used. Therefore, the metric \eqref{fjhjhgiikri} is geometrically regular across corner points.

Finally, consider a neighborhood $\mathcal{V}_p\subset\hat{M}^2$ of a pole, separating an axis rod $\Gamma_l$ to the north (without loss of generality) having rod structure $(1,0)$, and a horizon rod $\Gamma_h$ to the south. Similarly to the corner case, we will take the origin of the coordinate system to be centered at the pole. Then assuming analytic regularity and a sufficiently small $\mathcal{V}_p$, there exist regular functions $\bar{u}$ and $v$ with $u=\log(r-z)+\bar u$ such that
\begin{equation}
F=
\left( {\begin{array}{cc}
e^u & O(|r-z|) \\
O(|r-z|) & e^{v} \\
\end{array} } \right),
\end{equation}
and as in \eqref{2309i}, \eqref{2309i1} the potentials satisfy
\begin{equation}\label{2309ij10}
\zeta_1 +\frac{1}{3\sqrt{3}}\psi_{1}^2 \psi_2 +\psi_1 (\chi-b)=a_l +O(|r-z|^2),\quad\quad\quad
\zeta_2 -\frac{1}{3\sqrt{3}}\psi_1 \psi_2^2 
=\hat{a}_l +O(|r-z|),
\end{equation}
\begin{equation}\label{2309i1j10}
\chi+\frac{1}{\sqrt{3}}\psi_1 \psi_2=b_l +O(|r-z|),\quad\quad\quad
\psi_1=c_l +O(|r-z|),
\end{equation}
for some potential constants $a_l$, $\hat{a}_l$, $b_l$, and $c_l$. This implies that the potential expressions within the brackets of \eqref{alphaeq} are bounded. It follows that
\begin{align}
\begin{split}
\alpha_{\rho}=&-\frac{\rho}{2r^2}+\frac{z}{2r}\bar{u}_{\rho}
+\frac{(z-r)}{4r}\tilde{v}_{\rho}+\frac{\rho}{4r}(2\bar{u}_z +\tilde{v}_z)+O(\rho),\\
\alpha_z =&-\frac{z}{2r^2}+\frac{z}{2r}\bar{u}_{z}
+\frac{(z-r)}{4r}\tilde{v}_{z}-\frac{\rho}{4r}(2\bar{u}_{\rho} +\tilde{v}_{\rho})+O(\rho^2),
\end{split}
\end{align}
where $\tilde{v}=v+\tilde{f}$ and $e^{\tilde{f}}=e^{-u-v}f$. On $\mathcal{V}_p$ we then have
\begin{equation}
\alpha=-\frac{1}{2}\log r +\frac{z}{2r}(\bar{u}-\bar{u}(0))+\frac{(z-r)}{4r}(\tilde{v}-\tilde{v}(0))+c +O(\rho^2),
\end{equation}
for some constant $c$ and where $\bar{u}(0)$, $\tilde{v}(0)$ denote these functions evaluated at the pole.
The spacetime metric over $\mathcal{V}_p$ may then be written as
\begin{align}
\begin{split}
\mathbf{g}=&-\frac{\rho^2 e^{-\bar{u}-\tilde{v}}}{r-z}dt^2 +(r-z) e^{\bar{u}} (d\phi^1)^2 +e^v (d\phi^2)^2 +O(|r-z|)d\phi^1 d\phi^2\\
&+r^{-1}\exp\left(2c+\frac{z}{r}(\bar{u}-\bar{u}(0))
+\frac{(z-r)}{2r}(\tilde{v}-\tilde{v}(0))+O(\rho^2)\right)(d\rho^2 +dz^2).
\end{split}
\end{align}
The absence of a conical singularity on the axis rod implies that $2 e^{2c}=e^{\bar{u}(0)}$.
As above we are motivated to change to new coordinates $\xi,\eta\geq 0$ given by $z+i\rho=(\xi+i\eta)^2$, or equivalently $\rho=2\xi\eta$, $z=\xi^2-\eta^2$ in which the metric takes the form
\begin{align}
\begin{split}
\mathbf{g}=&-2\xi^2 e^{-\bar{u}-\tilde{v}} dt^2 +2\eta^2 e^{\bar{u}} (d\phi^1)^2 +e^v (d\phi^2)^2
+O(\eta^2)d\phi^1 d\phi^2\\
&+2 \exp\left(\bar{u}(0)+\frac{\xi^2(\bar{u}-\bar{u}(0))
-\eta^2(\bar{u}+\tilde{v}-\bar{u}(0)-\tilde{v}(0))}{\xi^2+\eta^2}
+O\left(\xi^2 \eta^2\right)\right)(d\xi^2 +d\eta^2).
\end{split}
\end{align}
Now define Cartesian coordinates $x=\eta \cos\phi^1$, $y=\eta \sin\phi^1$, and Kruskal-type coordinates $X,Y>0$ with $XY=\xi^2$, $X/Y=\exp\left(2te^{-\bar{u}(0)-\frac{1}{2}\tilde{v}(0)}\right)$ so that
\begin{equation}
d\eta^2 +\eta^2 (d\phi^1)^2=dx^2 +dy^2, \quad\quad \eta^2 d\eta^2=(x dx+ydy)^2,
\end{equation}
\begin{equation}
d\xi^2 -e^{-2\bar{u}(0)-\tilde{v}(0)}\xi^2 dt^2=dXdY,\quad\quad \xi^2 d\xi^2=\frac{1}{4}(YdX+XdY)^2.
\end{equation}
We then have
\begin{align}
\begin{split}
\mathbf{g}=& 2e^{-\bar{u}-\tilde{v}+2\bar{u}(0)+\tilde{v}(0)}\left(dX dY+O\left(\frac{2\bar{u}+\tilde{v}-2\bar{u}(0)-\tilde{v}(0)}{\xi^2 +\eta^2}+\eta^2\right)(YdX+XdY)^2 \right)\\
&+2e^{\bar{u}}\left(dx^2 +dy^2 +O\left(\frac{2\bar{u}+\tilde{v}-2\bar{u}(0)-\tilde{v}(0)}{\xi^2 +\eta^2}+\xi^2\right)(xdx +ydy)^2 \right)\\
&
+e^v (d\phi^2)^2 +O(1)(xdy-ydx)d\phi^2.
\end{split}
\end{align}
Since $2\bar{u}+\tilde{v}$ is a regular function at the pole we find that the error terms involving this quantity are regular. It follows that the metric \eqref{fjhjhgiikri} is geometrically regular across horizon poles.
 
A similar analysis may be used to show that the Maxwell field $\mathcal{F}$ of \eqref{fkiisjgjakng} is a regular 2-form on the spacetime, assuming analytic regularity of the harmonic map.

\section{The Coset Representative and Mazur Quantity}
\label{appendixb}

In this appendix we give an explicit expression for the positive definite unimodular coset representative $\Psi$ parameterizing the noncompact symmetric space target manifold $G_{2(2)}/SO(4)$, and use it to compute to the Mazur quantity as well as to establish positivities that are employed in Section \ref{sec7}. We will follow the presentation given in \cite{Kunduri:2011zr} which is adapted for a reduction of the supergravity equations with two spacelike Killing fields, rather than one timelike and one spacelike Killing field derived originally in~\cite{Bouchareb:2007ax}.  Let $\tilde{N}$ denote the set of $7\times 7$ positive definite unimodular matrices, then $\Psi: \mathbb{R}^3 \setminus \Gamma \to \tilde{N}$ is given by: 
\begin{equation}
\label{Psiexplicit}
\Psi= \left(
\begin{array}{ccc}  A & B & \sqrt{2} R \\
B^T & C & \sqrt{2} U \\
\sqrt{2} R^T & \sqrt{2} U^T & S \end{array} \right),
\end{equation} with inverse 
\begin{equation}
\label{Mexplicit1}
\Psi^{-1}= \left(
\begin{array}{ccc}  C & B^T & -\sqrt{2} U \\
B & A & -\sqrt{2} R\\
-\sqrt{2} U^T & -\sqrt{2} R^T & S \end{array} \right),
\end{equation}
where $A$, $C$ are symmetric $3\times 3$ matrices, $B$ is a $3\times 3$ matrix, $R$, $U$ are $3\times 1$ matrices and $S$ is a scalar. Upon setting $\chi = \sqrt{3}\mu$, $\psi_i = -\sqrt{3} \nu_i$, and $F = (f_{ij})$,  $f = \det F $ these submatrices may be expressed as
\begin{align}
\begin{split}
A =& \left( \begin{array}{cc}
(1+F^{-1} \mu^2)F + F^{-1} \tilde{\zeta} \tilde{\zeta}^T +(2+ \nu^T \nu)\nu \nu^T+\frac{\mu}{\sqrt{f}} (\nu \nu^T \mathbf{J} -\mathbf{J}\nu \nu^T ) & - f^{-1} \tilde{\zeta} \\
-f^{-1}\tilde{\zeta}^T & f^{-1} \end{array} \right) ,  \\
B =& \left(\begin{array}{cc} \nu \nu^T  - \frac{\mu}{\sqrt{f}}\mathbf{J} + \frac{1}{\sqrt{f}}\tilde{\zeta} \nu^T \mathbf{J} & \upsilon \\
-\frac{\nu^T \mathbf{J}}{\sqrt{f}} & \frac{\mu^2}{f} - \frac{\nu^T \mathbf{J} \tilde{\zeta}}{\sqrt{f}}
\end{array}\right),
\\
C =& \left( \begin{array}{cc}   (1+\nu^T \nu)F^{-1} -\nu \nu^T & \tilde{\zeta} +(\frac{\mu^2}{\sqrt{f}} - \nu^T \mathbf{J} \tilde{\zeta})\mathbf{J}\nu -\mu \nu \\  \tilde{\zeta}^T +(\frac{\mu^2}{\sqrt{f}} - \nu^T \mathbf{J}\tilde{\zeta})\mathbf{J}\nu^T -\mu \nu^T   & c
\end{array} \right),\\
U=& \left( \begin{array}{c} (F^{-1} - \frac{\mu}{\sqrt{f}}\mathbf{J})\nu \\ \nu^T F^{-1}\tilde{\zeta} - \mu[1 + \nu^T \nu + \frac{\mu^2}{f} - \frac{1}{\sqrt{f}}\nu^T \mathbf{J} \tilde{\zeta}] \end{array} \right), \\
R=& \left( \begin{array}{c} (1+\nu^T \nu)\nu - \frac{\mu}{\sqrt{f}}\mathbf{J}\nu + \frac{\mu}{f}\tilde{\zeta} \\
-\frac{\mu}{f} \end{array}\right),  \\
S =& 1 + 2(\nu^T \nu + f^{-1}\mu^2), 
\end{split}
\end{align}
where $\nu$ represents a column vector with components $\nu_i$ and similarly for $\zeta$ with
\begin{align}
\begin{split}
\tilde{\zeta}_i =&\zeta_i-\mu \nu_i , \\
\upsilon =& -\left(1 - \frac{\mu^2}{f}\right)\sqrt{f}\mathbf{J}\nu  - (2 + \nu^T \nu)\mu \nu + \nu^T F^{-1}\tilde{\zeta}\nu  + \left(-\frac{\mu^2}{f} + \frac{\nu^T \mathbf{J}\tilde{\zeta}}{\sqrt{f}}\right)\tilde{\zeta}  - \frac{\mu}{\sqrt{f}}\mathbf{J}\tilde{\zeta} ,\\
c =& \tilde{\zeta}^T\tilde{\zeta} - 2\mu \nu^T\tilde{\zeta} +f[1+ \nu^T \nu+(2+\nu^T \nu)f^{-1}\mu^2 +f^{-2}(\mu^2- \nu^T \mathbf{J} \tilde{\zeta} \sqrt{f} )^2],
\end{split}
\end{align} 
and
\begin{equation}
\mathbf{J}= \left(
\begin{array}{cc}  0 & 1 \\
-1 & 0  \end{array} \right).
\end{equation} 
For comparison, when the Maxwell field is set to zero ($\chi=\psi_i =0$), the above coset representative considerably simplifies to 
\begin{equation}
\Psi_{\text{vacuum}} = \begin{pmatrix} \Phi^{-1} & 0 & 0 \\ 0 & \Phi & 0 \\ 0 & 0 & 1 \end{pmatrix},
\end{equation} 
where $\Phi$ is the coset representative for the target space $SL(3, \mathbb{R}) / SO(3)$ used in \cite{KhuriWeinsteinYamada}.

We are interested in computing the {\it Mazur quantity} $\text{Tr} (\Psi_0^{-1} \Psi)$, where $\Psi$ and $\Psi_0$ are two coset representatives.  In order to carry out this computation, it is beneficial to express \eqref{Psiexplicit} as a product of simpler matrices. Such a decomposition is given by Clement~\cite{Clement:2007qy}, with the replacement $\tau \to -f$ to take into account the fact that we are reducing on two spacelike Killing fields. We will write 
\begin{equation}\label{Mazurquantity}
\Psi=\mathcal{V}_{\zeta}^T\mathcal{V}_{\mu}^T\mathcal{V}_{\nu}^T \mathfrak{F}\mathcal{V}_{\nu}\mathcal{V}_{\mu}\mathcal{V}_{\zeta},
\end{equation}
where $\mathcal{V}_{\nu}, \mathcal{V}_{\mu}, \mathcal{V}_{\zeta},$ and $\mathfrak{F}$ are the following unimodular $7\times 7$ matrices, which for notational convenience are expressed as $5 \times 5$ matrices with the first and third columns each represent two columns (e.g. the upper left minor in $\mathfrak{F}$ and $\mathcal{V}_\zeta$ are the $2\times 2$ matrices $F$ and the 2-dimensional identity matrix $I_2$ respectively):
\begin{equation}
{\mathfrak{F}}= \left(
\begin{array}{ccccc}  F & 0 & 0 &0& 0 \\
0 & f^{-1} & 0 &0& 0 \\
0 & 0 & F^{-1} &0& 0 \\
0 & 0 & 0 &f& 0 \\
0 & 0 & 0 &0& 1 
\end{array} \right),\qquad 
\mathcal{V}_{\zeta}= \left(
\begin{array}{ccccc}  I_2 & 0 & 0 &0& 0 \\
-\zeta^T & 1& 0 &0& 0 \\
0 & 0 & I_2 &\zeta& 0 \\
0 & 0 & 0 &1& 0 \\
0 & 0 & 0 &0& 1 
\end{array} \right),
\end{equation}
\begin{equation}
\mathcal{V}_{\mu}= \left(
\begin{array}{ccccc}  
I_2 & 0 & 0 &0& 0 \\
0 & 1 & 0 &\mu^2& -\sqrt{2}\mu \\
\mu \mathbf{J} & 0 & I_2 &0& 0 \\
0 & 0 & 0 &1& 0 \\
0 & 0 & 0 &-\sqrt{2}\mu& 1 
\end{array} \right),\qquad 
\mathcal{V}_{\nu}= \left(
\begin{array}{ccccc}  I_2 & 0 & 0 &-\mathbf{J}\nu& 0 \\
0 & 1& -\nu^T \mathbf{J} &0& 0 \\
\nu\nu^T & 0 & I_2 &0& \sqrt{2}\nu \\
0 & 0 & 0 &1& 0 \\
\sqrt{2}\nu^T & 0 & 0 &0& 1 
\end{array} \right).
\end{equation}   
Let $P=\mathcal{V}_{\nu}\mathcal{V}_{\mu}\mathcal{V}_{\zeta}$ and observe that
\begin{equation}
\text{Tr}\left(\Psi^{-1}_0\Psi\right)=\text{Tr}\left(P_0^{-1}\mathfrak{F}^{-1}_0(P^T_0)^{-1}P^T \mathfrak{F}P\right)
=\text{Tr}\left(\mathfrak{F}^{-1}_0(PP_0^{-1})^T \mathfrak{F}PP_0^{-1}\right),
\end{equation} 
where the subscript $0$ indicates matrices associated with $\Psi_0$.  It follows that
\begin{equation}
PP_0^{-1}= \left(
\begin{array}{ccccc}  I_2 & 0 & 0 &-\mathbf{J}(\nu-\nu_0)& 0 \\
L_1 & 1& -(\nu-\nu_0)^T \mathbf{J} &L_2& L_3 \\
L_5 & 0 & I_2 &L_4& \sqrt{2}(\nu-\nu_0) \\
0 & 0 & 0 &1& 0 \\
\sqrt{2}(\nu-\nu_0)^T & 0 & 0 &-\sqrt{2}(\mu-\mu_0)-\sqrt{2}\nu_0^T \mathbf{J}\nu& 1 
\end{array} \right),
\end{equation}
where we have defined the (row) vector
\begin{equation}
L_1=2\nu_0^T(\mu-\mu_0)-(\zeta-\zeta_0)^T+\nu_0^T(\nu_0^T \mathbf{J}\nu)+\nu^T(\mu-\mu_0),
\end{equation} 
the scalars 
\begin{align}
\begin{split}
L_2&=(\mu-\mu_0)^2+(\zeta-\zeta_0)^T \mathbf{J}(\nu-\nu_0) - (\mu-\mu_0)\nu_0^T \mathbf{J}\nu ,\\
L_3&=-\sqrt{2}(\mu-\mu_0)-\sqrt{2}\nu_0^T \mathbf{J}\nu,
\end{split}
\end{align} 
the (column) vector
\begin{equation}
L_4=(\nu^T \mathbf{J}\nu_0)\nu -(\mu-\mu_0)\nu_0 -2(\mu-\mu_0)\nu+(\zeta-\zeta_0),
\end{equation} 
and finally the $2\times 2$ matrix
\begin{equation}
L_5=(\nu-\nu_0)(\nu-\nu_0)^T+\left[\nu_0^T \mathbf{J}\nu+(\mu-\mu_0)\right]\mathbf{J}.
\end{equation} 
Using this decomposition it is straightforward to calculate
\begin{equation}\label{right}
\footnotesize \mathfrak{F}PP_0^{-1}= \left(
\begin{array}{ccccc}  
F & 0 & 0 &-F \mathbf{J}(\nu-\nu_0)& 0 \\
f^{-1}L_1 & f^{-1}& -f^{-1}(\nu-\nu_0)^T \mathbf{J} &f^{-1}L_2& f^{-1}L_3 \\
F^{-1}L_5& 0 & F^{-1} &F^{-1}L_4& \sqrt{2}F^{-1}(\nu-\nu_0) \\
0 & 0 & 0 &f& 0 \\
\sqrt{2}(\nu-\nu_0)^T & 0 & 0 &-\sqrt{2}(\mu-\mu_0)-\sqrt{2}\nu_0^T \mathbf{J}\nu& 1 
\end{array} \right),
\end{equation} 
and
\begin{equation}\label{left}
\footnotesize  \mathfrak{F}_0^{-1}(PP_0^{-1})^T=\left(
\begin{array}{ccccc}  
F_0^{-1} & F_0^{-1} L_1^T & F^{-1}_0L_5^T &0& \sqrt{2}F^{-1}_0(\nu-\nu_0)\\
0 & f_0& 0&0& 0 \\
0 &  F_0 \mathbf{J}(\nu-\nu_0) & F_0 &0&0 \\
f_0^{-1}(\nu-\nu_0)^T \mathbf{J} & f_0^{-1}L_2 & f_0^{-1}L_4^T &f_0^{-1}& -\sqrt{2}f_0^{-1}(\mu-\mu_0)-\sqrt{2}f_0^{-1}\nu_0^T \mathbf{J}\nu \\
0& L_3 & \sqrt{2}(\nu-\nu_0)^T &0& 1 
\end{array} \right) . 
\end{equation} 
Finally, the Mazur quantity is obtained by taking the trace of the matrix resulting from the multiplication of \eqref{right} and \eqref{left}, which yields
\begin{align}
\begin{split}
\text{Tr}\left(\Psi^{-1}_0\Psi\right) =& \text{Tr}(F_0^{-1} F) + f^{-1} \text{Tr}( F_0^{-1} L_1^T L_1) +  \text{Tr}(F_0^{-1} L_5^T F^{-1} L_5) \\ 
&+ 2  \text{Tr}[F_0^{-1} (\nu - \nu_0)(\nu-\nu_0)^T] + f_0f^{-1} +  f^{-1} \text{Tr} [F_0 \mathbf{J} (\nu - \nu_0) (\mathbf{J}(\nu - \nu_0))^T ] \\  
&+ \text{Tr}(F_0 F^{-1})  + f_0^{-1} (\mathbf{J}(\nu - \nu_0))^T F \mathbf{J} (\nu - \nu_0) + f_0^{-1} f^{-1} L_2^2+ f_0^{-1} L_4^T F^{-1} L_4  \\ 
&+ f_0^{-1} f + 2f_0^{-1} \left(\mu - \mu_0 + \nu_0^T \mathbf{J} \nu\right)^2 + f^{-1} L_3^2+ 2 (\nu - \nu_0)^T F^{-1} (\nu - \nu_0) + 1.
\end{split}
\end{align} 
Since the matrices $F$ and $F_0$ are positive semi-definite, it may be verified that each term in this expression is nonnegative, with the help of the elementary fact that the product of two positive semi-definite matrices has nonnegative trace.





\end{document}